\documentclass[12pt,onecolumn]{IEEEtran}
\usepackage{amsmath,amsfonts}
\usepackage{amssymb}
\pdfoutput=1

\newenvironment{proof}{{\it Proof:}}{\hfill$\blacksquare$}

\newtheorem{lemma}{Lemma}
\newtheorem{definition}{Definition}
\usepackage{algorithmic}
\usepackage{float}
\usepackage{array}
\usepackage[caption=false,font=normalsize,labelfont=sf,textfont=sf]{subfig}
\usepackage{textcomp}
\usepackage{stfloats}
\usepackage{url}
\usepackage{verbatim}
\usepackage[ruled,vlined]{algorithm2e}
\usepackage{graphicx}
\usepackage{cite}
\usepackage{color}
\usepackage[colorlinks = true,
            linkcolor = blue,
            urlcolor  = blue,
            citecolor = blue,
            anchorcolor = blue]{hyperref}

\begin{document}

\title{Joint RIS-UE Association and Beamforming Design in RIS-Assisted Cell-Free MIMO Network }

\author{ Hongqin Ke, Jindan Xu,~\IEEEmembership{Member,~IEEE,} Wei Xu,~\IEEEmembership{Fellow,~IEEE,} \\ Chau Yuen,~\IEEEmembership{Fellow,~IEEE,} and Zhaohua Lu
        % <-this % stops a space
\thanks{Hongqin Ke and Wei Xu are with the National Mobile Communications Research Lab, Southeast University, Nanjing 210096, China; Wei Xu is also with Purple Mountain Laboratories, Nanjing 211111, China (e-mail: {kehongqin, wxu}@seu.edu.cn).}% <-this % stops a space
\thanks{Jindan Xu and Chau Yuen are with the School of Electrical and Electronics Engineering, Nanyang Technological University, Singapore 639798, Singapore (e-mail: jindan.xu@ntu.edu.sg, chau.yuen@ntu.edu.sg).}
\thanks{Zhaohua Lu is with ZTE Corporation, Shenzhen 122008, China, and also with the State Key Laboratory of Mobile Network and Mobile Multimedia Technology, Shenzhen 518057, China (e-mail: lu.zhaohua@zte.com.cn).}}

\markboth{}%
{Shell \MakeLowercase{\textit{et al.}}: A Sample Article Using IEEEtran.cls for IEEE Journals}


\maketitle

\begin{abstract}
\fontsize{12}{15}\selectfont
Reconfigurable intelligent surface (RIS)-assisted cell-free (CF) multiple-input multiple-output (MIMO) networks can significantly enhance system performance. However, the extensive deployment of RIS elements imposes considerable channel acquisition overhead, with the high density of nodes and antennas in RIS-assisted CF networks amplifying this challenge. To tackle this issue, in this paper, we explore integrating  RIS-user equipment (UE) association into downlink RIS-assisted CF transmitter design, which greatly reduces the channel acquisition costs. The key point is that once UEs are associated with specific RISs, there is no need to frequently acquire channels from non-associated RISs. Then, we formulate the problem of joint RIS-UE association and beamforming at APs and RISs to maximize the weighted sum rate (WSR). In particular, we propose a two-stage framework to solve it. In the first stage, we apply a many-to-many matching algorithm to establish the RIS-UE association. In the second stage, we introduce a sequential optimization-based method that decomposes the joint optimization of RIS phase shifts and AP beamforming into two distinct subproblems. To optimize the RIS phase shifts, we employ the majorization-minimization (MM) algorithm to obtain a semi-closed-form solution. For AP beamforming, we develop a joint block diagonalization algorithm, which yields a closed-form solution. Simulation results demonstrate the effectiveness of the proposed algorithm and show that, while RIS-UE association significantly reduces overhead, it incurs a minor performance loss that remains within an acceptable range. Additionally, we investigate the impact of RIS deployment and conclude that RISs exhibit enhanced performance when positioned between APs and UEs.
 
\end{abstract}

\begin{IEEEkeywords}
\fontsize{12}{15}\selectfont
Cell-free (CF) network, reconfigurable intelligent surface (RIS),  beamforming design, RIS-UE association, resource allocation.
\end{IEEEkeywords}

\section{Introduction}
\IEEEPARstart{D}{riven} by promising advancements such as massive multiple-input multiple-output (mMIMO) and ultra-dense networks (UDN) \cite{5G}, next-generation wireless communication networks are anticipated to deliver not only higher data rates but also reduced latency, enhanced reliability, and greater connectivity to support emerging applications \cite{MIMOnew, AI}. Achieving these goals necessitates the deployment of a substantial number of co-located antennas and base stations to enable mMIMO and UDN, though this approach entails significant costs and increased power consumption. Furthermore, as cell density rises, inter-cell interference becomes more pronounced, exerting a serious impact on system performance \cite{9064545}. 

Fortunately, cell-free (CF) MIMO, a novel user-centric network architecture, has been proposed to address these challenges \cite{9586055, ZTE}. Unlike conventional ultra-dense networks (UDNs), CF MIMO systems employ geographically distributed access points (APs) that collaborate to serve all user equipment (UEs) through coherent transmission, with all APs linked to a central processing unit (CPU) via fiber optics or wireless connections. Owing to its robust interference management, reduced deployment costs, and enhanced macro diversity, CF MIMO can significantly improve network coverage and system capacity \cite{cellfree3}. However, realizing the full performance potential of CF MIMO requires a large number of distributed APs, which, in turn, escalates both costs and power consumption—challenges reminiscent of those faced in UDNs. These concerns have become even more pressing given the increasing focus on environmentally sustainable and green communication principles.

The reconfigurable intelligent surface (RIS) has emerged as a promising technique among various candidates due to its unique attributes of low cost, low energy consumption, and programmability \cite{RIS1, xu2024RIS, xu2023reconfiguring}. With the aid of an intelligent controller, an RIS, composed of multiple passive elements, enhances communication by reprogramming incident signals from the AP and reflecting them toward the UE in a specified direction. Additionally, RISs can be easily integrated into existing communication scenarios and applications, offering the advantage of broad deployability due to their affordability and low power requirements \cite{KHQ_ris}.

Building on the discussions above, the RIS-assisted CF MIMO system has sparked a surge of research interest due to the distinct advantages offered by RISs \cite{shi2023ris}. Numerous studies have focused on optimizing beamforming designs for both APs and RISs, evaluating their impact across various communication metrics, such as energy efficiency (EE) and spectral efficiency (SE). For example, the global maximization of EE has been extensively investigated \cite{energy1, energy2, energy3}. In particular, \cite{energy1} addresses a power consumption model tailored for discrete RIS phase shifts, \cite{energy2} considers constraints arising from limited backhaul capacity, and \cite{energy3} integrates an energy consumption model suited for wideband systems. On another front, \cite{twotimescale} examines network capacity maximization through a two-timescale scheme designed to minimize overhead and computational complexity, while \cite{yao_robust_2023} studies the maximization of the worst-case sum rate under uncertain channel state information (CSI), especially in scenarios with constrained backhaul capacity. Beyond these optimization-oriented studies, recent research has also investigated innovative applications of RIS-assisted CF systems, including unmanned aerial vehicles (UAVs) \cite{UAV1, UAV2}, physical layer security (PLS) \cite{PLS}, and wireless energy transfer (WET) \cite{WET}. For instance, \cite{UAV2} proposes leveraging RIS to enhance UAV communication while maintaining or even improving the downlink rate for ground UEs in the CF network. Likewise, in \cite{PLS}, the secrecy performance of the system is analyzed in the presence of multiple active eavesdroppers. Taken together, these findings, validated through both theoretical analysis and simulation, suggest that integrating RIS with CF architectures can yield significant and reliable performance gains.

It is worth noting that the majority of the aforementioned studies, including \cite{energy1,energy2,energy3,yao_robust_2023,UAV1,UAV2,PLS,WET}, assume an idealized scenario where multiple RISs serve each UE simultaneously. However, this approach proves impractical, as the large number of RIS elements results in increased channel dimensions and significantly escalates the signaling overhead required for channel acquisition \cite{RIS_es1}.
 To mitigate channel acquisition overhead, implementing an association strategy between the RIS and the UE can be advantageous. Specifically, when channel conditions between an RIS and a UE are poor—such as when they are geographically distant or obstructed by significant barriers—the system gains minimal performance enhancement from the RIS. In such cases, establishing an RIS-UE association reduces the necessity for frequent channel acquisition for non-associated RIS-UE pairs, thereby optimizing the balance between system performance and signaling overhead. 

However, incorporating RIS-UE association into RIS-assisted CF systems presents several critical challenges, including the design of RIS-UE association and joint beamforming for APs and RISs. First, establishing associations between multiple RISs and multiple UEs requires careful consideration. In existing studies on RIS-UE association, \cite{twotimescale} proposes a linear conic relaxation algorithm to establish a many-to-one matching between multiple RISs and a single UE, while \cite{RISnum} employs a graph-theory-based approach for one-to-one matching between a RIS and a UE. Evidently, the concept of many-to-many matching between multiple RISs and UEs remains underexplored in RIS-assisted CF systems. More critically, the presence of numerous nodes and antennas in RIS-assisted CF systems introduces significant complexity to joint beamforming. Moreover, previous studies have largely overlooked a comprehensive beamforming design that accommodates multiple APs, RISs, and UEs, each with their respective multiple antennas \cite{energy1,energy2,energy3,yao_robust_2023,UAV1,UAV2,PLS,WET}. These considerations underscore the need for a more generalized and efficient beamforming design framework.

Inspired by these discussions, in this paper, we focus on improving the weighted sum rate (WSR) of a RIS-assisted CF MIMO system by jointly optimizing the RIS-UE association, the AP beamforming, and the RIS phase shifts. The main contributions of this paper are as follows:
\begin{itemize}
    \item {\textit{A general architecture of  RIS-assisted CF network with RIS-UE association:} This paper investigates a generalized downlink communication framework that leverages multiple RISs to assist in CF MIMO, encompassing multiple APs, RISs, UEs, and antennas to broaden the system’s applicability. To manage the significant channel acquisition overhead introduced by the large number of RIS elements, we implement an RIS-UE association strategy that sets an upper limit on the number of UEs each RIS can serve and the number of RISs each UE can connect to, enhancing the system’s practicality. Under the constraints of AP power, RIS element unit modulus, and RIS-UE matching, we formulate the maximization of the weighted sum rate as a mixed-integer nonlinear programming (MINLP) problem, jointly optimizing RIS-UE association, RIS phase shifts, and AP transmit beamforming.}
    \item {\textit{A low-complexity solution with closed-form expressions:} To address the inherent complexity of the MINLP problem, we propose an innovative two-stage framework that decomposes the original problem into several tractable subproblems. Specifically, in the first stage, we model RIS-UE association as a many-to-many matching problem and introduce an efficient, low-complexity algorithm to establish these associations. Additionally, we develop an Majorization-Minimization (MM)-based approach to manage the non-convex unit-modulus constraints in RIS phase shift design, with the MM algorithm applied iteratively in both the first and second stages. In the second stage, leveraging the optimal RIS-UE association and RIS phase shifts, we derive the optimal AP beamforming solution by joint block diagonalization (BD) with a bisection search. Notably, both the MM and BD algorithms yield closed-form solutions, which substantially enhances computational efficiency and enables rapid convergence. This framework not only provides a structured solution to a highly complex problem but also demonstrates practical applicability in RIS-assisted CF MIMO systems.}
    \item {\textit{Performance validation and analysis:} Comparisons with various benchmark schemes validate the effectiveness of the proposed algorithm for RIS-assisted CF MIMO systems. Simulation results demonstrate that while adopting RIS-UE association in RIS-assisted CF systems entails a minor performance trade-off, it substantially reduces channel acquisition overhead and remains within an acceptable performance loss range. Additionally, our scheme exhibits strong robustness against CSI errors. Finally, we offer insights and recommendations for optimal RIS deployment within the considered system. }
\end{itemize}

The remainder of this paper is organized as follows. In Section \ref{sectwo}, we describe the system model and the weighted sum rate maximization problem under AP power constraints, RIS modulus constraints, and RIS-UE matching constraints. In Section \ref{secthree}, we propose a two-stage framework to solve the optimization problem. In Section \ref{secfour}, we verify the convergence of the proposed algorithm and analyze the performance of the RIS-assisted CF MIMO system. Finally, our conclusions are given in Section \ref{secfive}.

\textit{Notations}: $a$, $\mathbf{a}$, and $\mathbf{A}$ denote scalar, vector, and matrix respectively. $\mathrm{Re} (\cdot )$ means to take the real part. ${\bf I}_N$ denotes an $N\times N$-dimensional identity matrix.  $\mathrm{diag} (\cdot )$ returns the diagonal matrix of the input vector and  $\mathrm{blkdiag} (\mathbf{A}_1,\dots,\mathbf{A}_n  )$ returns a block diagonal matrix created by $\mathbf{A}_1,\dots,\mathbf{A}_n$. $\lceil\cdot\rceil$ denotes upward rounding and $\otimes $ denotes Kronecker product. $\mathbf{A}^{\mathrm{\mathit{T} } } ,\mathbf{A}^{\mathrm{\mathit{H} } }, \mathbf{A}^{\mathrm{-1} },\mathbf{A}^{\mathrm{1/2}},\mathrm{Rank} (\mathbf{A}),\mathrm{Tr} (\mathbf{A})$ denote transpose, conjugate transpose, inverse, square root, rank, and trace of $\mathbf{A}$ respectively. $\left \| \cdot  \right \|_{F} $ is the Frobenius norm of matrix. $\left | \cdot  \right | $ means the length of a set or the determinant of a matrix. $\mathcal{C N}\left(\mathbf{0}, \mathbf{\Sigma} \right)$ represents the complex Gaussian distribution with zero mean and variance matrix.

\section{System Model}  \label{sectwo}

\begin{figure}[h]
    \centering
\includegraphics[width=0.7\columnwidth,height=0.49\linewidth]{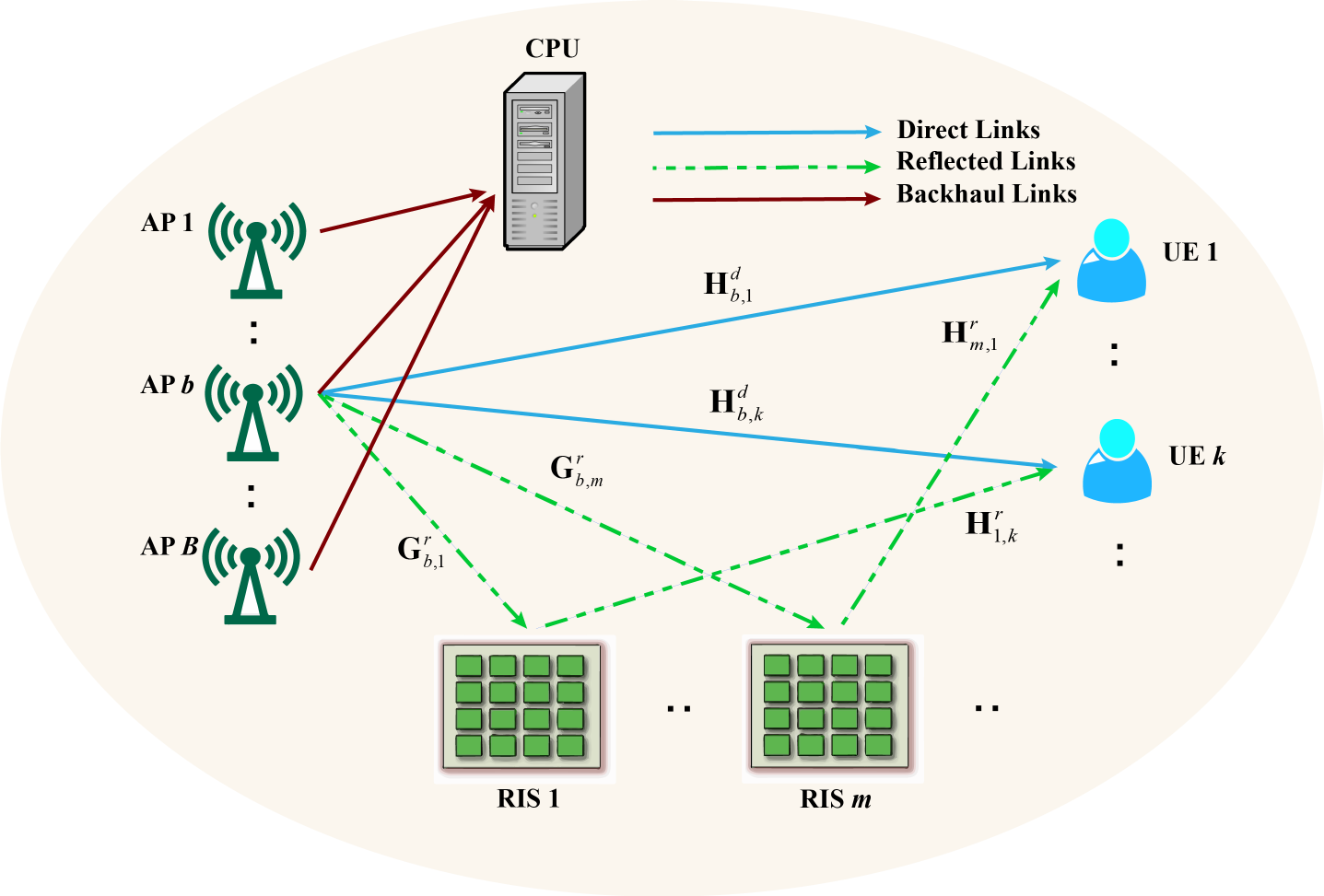}
    \caption{The downlink RIS-assisted CF-MIMO network}
    \label{system_fig}
\end{figure}
\subsection{Signal Model}

As depicted in Fig. \ref{system_fig}, we consider an RIS-assisted downlink cell-free system comprising $K$ multiple-antenna UEs, $B$ multiple-antenna APs, and $M$ RISs, where all APs and RISs are connected to the central processing unit (CPU) through  wireless backhaul or optical fiber. We define $\mathcal{B}=\left\{1,\ldots,B\right\}$, $\mathcal{K}=\left\{1,\ldots,K\right\}$, and $\mathcal{M}=\left\{1,\ldots,M\right\}$ as the index sets of APs, UEs, and RISs, respectively. Each RIS contains $N$ elements, while each AP and UE is equipped with $N_t$ and $N_r$ antennas, respectively. 

In the considered system, all APs cooperate to serve all UEs over the same time-frequency resources. {For simplicity, we assume that the total number of antennas in the UEs does not exceed that in the APs, i.e., ${KN_r } \le{BN_t}$, to enable effective spatial multiplexing by the APs and ensure that each UE can achieve multi-stream transmission. If the system is overloaded, the APs can mitigate this by dynamically adjusting resource allocation, assigning UEs to different time- or frequency-resource blocks\cite{overload1}, or grouping nearly orthogonal UEs together for simultaneous data transmission\cite{overload2}, to meet the above assumption.}

We introduce a binary variable $c_{m,k}\in\left\{0,1\right\}$ to indicate whether the $m$-th RIS and $k$-th UE are matched, where $c_{m,k}=1$ signifies that the $m$-th RIS serves the $k$-th UE, and $c_{m,k}=0$ otherwise. {We assume that each RIS can serve up to $U_{\mathrm{match}}$ UEs, and each UE can be matched with a maximum of $R_{\mathrm{match}}$ RISs during the association process. These constraints can be expressed as  $U_{\mathrm{match}} \geq \sum_{k \in \mathcal{K}} c_{m, k} , \ \forall m \in \mathcal{M}$, and $R_{\mathrm{match}} \geq \sum_{m \in \mathcal{M}} c_{m, k} , \ \forall k \in \mathcal{K}$.
}

In a cell-free MIMO network with numerous geographically distributed APs, signals from different APs typically experience varying propagation delays, necessitating synchronization among all APs to enable coherent transmission. 
The signal transmitted by the $b$-th AP is given by: 
\begin{equation}
\mathbf{x}_b=\sum_{k=1}^{K}\mathbf{F}_{b,k}\mathbf{s}_k,
\end{equation}
where $\mathbf{s}_{k}\in\mathbb{C}^{N_s\times1}$ is the symbol transmitted to the $k$-th UE, satisfying $\mathbb{E}\left[\mathbf{s}_{k} \mathbf{s}_{i}^{H}\right]= \mathbf{0},\forall k\neq i$, and $\mathbb{E}\left[\mathbf{s}_{k} \mathbf{s}_{k}^{\mathit{H} }\right]=\mathbf{I}_{N_{s}}$. $N_s$ denotes the number of data streams for each UE, and
$\mathbf{F}_{b,k} \in \mathbb{C}^{N_{t} \times N_{s}}$ is the linear precoding matrix used by the $b$-th AP for the $k$-th UE. Let $\mathbf{H}_{b,k}^d\in\mathbb{C}^{N_r\times N_t}$, $\mathbf{H}_{m,k}^r\in\mathbb{C}^{N_r\times N}$, and $\mathbf{G}_{b,m}^r\in\mathbb{C}^{N\times N_t}$ represent the channels between the $b$-th AP and $k$-th UE, the $m$-th RIS and $k$-th UE, and the $b$-th AP and $m$-th RIS, respectively. 

The received signal at the $k$-th UE can be written as 
\begin{equation}
%\label{deqn_ex1a}
\mathbf{y}_k=\sum_{b=1}^{B}\mathbf{H}_{b,k}^d\mathbf{x}_b+\sum_{b=1}^{B}{\sum_{m=1}^{M}{c_{m,k}\mathbf{H}_{m,k}^r\mathbf{\Phi}_m}\mathbf{G}_{b,m}^r}\mathbf{x}_b+\mathbf{n}_k,
\end{equation}
where $\mathbf{n}_k \sim \mathcal{CN}\left(\mathbf{0},\sigma^2\mathbf{I}_{N_r}\right)$ represents the additive white Gaussian noise with zero mean and variance $\sigma^2\mathbf{I}_{N_r}$, and $\mathbf{\Phi}_m\in\mathbb{C}^{N\times N}$ is the phase shift matrix of the $m$-th RIS, which is expressed as 
\begin{equation}
\mathbf{\Phi}_m=\mathrm{diag}\left(\phi_{m,1},\cdots,\phi_{m,n},\cdots,\phi_{m,N}\right),
\end{equation}
where $\phi_{m,n}=e^{j\theta_{m,n}},\forall m,n$, and $0\le\theta_{m,n}\le2\pi$. 

To facilitate further processing, define $\mathbf{H}_k^d=\left[\mathbf{H}_{1,k}^d,\ldots,\mathbf{H}_{B,k}^d\right]\in\mathbb{C}^{N_r\times B N_t}$ as the direct channel from all APs to the $k$-th UE, and $\mathbf{G}_m^r=\left[\mathbf{G}_{1,m}^r,\ldots,\mathbf{G}_{B,m}^r\right]\in\mathbb{C}^{N\times B N_t}$ as  the aggregate channel from all APs to the $m$-th RIS.
The equivalent channel from all APs to the $k$-th UE can then be expressed as: 
\begin{equation}
    \bar{\mathbf{H}}_{k} = \mathbf{H}_{k}^{d}+\sum_{m=1}^{M} c_{m,k}\mathbf{H}_{m, k}^{r} \boldsymbol{\Phi}_{m} \mathbf{G}_{m}^{r},
\end{equation}
where $\bar{\mathbf{H}}_{k} \in \mathbb{C}^{N_{r} \times B N_{t}}$. 

Then, the received signal at the $k$-th UE can be rewritten as 
\begin{equation}
\mathbf{y}_{k} = \underbrace{\bar{\mathbf{H}}_{k} \mathbf{F}_{k} \mathbf{s}_{k}}_{\text {Desired signal }} \\
+\underbrace{\sum_{i \neq k}^{K} \bar{\mathbf{H}}_{k} \mathbf{F}_{i} \mathbf{s}_{i}}_{\text {Intra interference }}+\underbrace{\mathbf{n}_{k}}_{\text {Noise }},
\end{equation}
which consists of the desired signal, intra interference, and noise components, where $\mathbf{F}_{k}=\left[\mathbf{F}_{1, k}^{\mathrm{\mathit{H} }}, \ldots, \mathbf{F}_{B, k}^{\mathrm{\mathit{H} }}\right]^{\mathrm{\mathit{H} }} \in \mathbb{C}^{B N_{t} \times N_{s}}$. We also define the transmission covariance matrix for the $k$-th UE as $\mathbf{W}_{k}=\mathbf{F}_{k} \mathbf{F}_{k}^{\mathrm{\mathit{H} }} \in \mathbb{C}^{B N_{t} \times B N_{t}}$, where it is evident that $\mathbf{W}_{k} \succeq 0$.

For any AP $b$ with maximum transmit power $\mathrm{P}_{b, \max }$, the power constraint must be satisfied, i.e., $\sum_{k=1}^{K} \operatorname{Tr}\left(\mathbf{T}_{b} \mathbf{W}_{k}\right) \leq \mathrm{P}_{b, \max }$, where
\begin{equation}
\mathbf{T}_{b} = \operatorname{diag}(\underbrace{0, \ldots, 0}_{(b-1) N_{t}}, \underbrace{1, \ldots,1}_{N_{t}}, \underbrace{0, \ldots, 0}_{(B-b) N_{t}}),
\end{equation}
and $\mathbf{T}_{b}$ is a binary diagonal matrix to extract the beamforming matrices of the $b$-th AP from $\mathbf{W}_{k}$.
The achievable rate for the $k$-th UE can then be expressed as
\begin{equation}
R_{k}(\mathbf{F}, \mathbf{\Phi},\mathbf{c }) =\log \left|\mathbf{I}+\bar{\mathbf{H}}_{k} \mathbf{F}_{k} \mathbf{F}_{k}^{\mathrm{\mathit{H} }} \bar{\mathbf{H}}_{k}^{\mathrm{\mathit{H} }} (\mathbf{J}_{k} + \sigma^{2} \mathbf{I})^{-1}\right|,
\end{equation}
where $\mathbf{F}=\left[\mathbf{F}_{ k}, \forall  k\right] ,\mathbf{\Phi }=\left[\mathbf{\Phi}_{ m}, \forall m\right],\mathbf{c }=\left[c_{m,k}, \forall m,k\right]$, and $\mathbf{J}_{k}=\sum_{i \neq k}^{K} \bar{\mathbf{H}}_{k} \mathbf{F}_{i} \mathbf{F}_{i}^{\mathrm{\mathit{H} }} \bar{\mathbf{H}}_{k}^{\mathrm{\mathit{H} }}$ represents the inter-user interference covariance matrix.

\subsection{Problem Formulation}
This subsection formulates the problem of maximizing the weighted sum rate through joint optimization of AP beamforming, RIS phase shifts, and RIS-UE association variables, expressed as follows:
\begin{align}
\label{X1}\max _{\mathbf{F}, \boldsymbol{\Phi},\mathbf{c}} & \ \ \sum_{k=1}^{K} \omega_{k} \log \left|\mathbf{I}+\bar{\mathbf{H}}_{k} \mathbf{F}_{k} \mathbf{F}_{k}^{\mathrm{\mathit{H} }} \bar{\mathbf{H}}_{k}^{\mathrm{\mathit{H} }} (\mathbf{J}_{k} + \sigma^{2} \mathbf{I})^{-1}\right| \\
\text { s.t. } & \ \ \sum_{k=1}^{K} \operatorname{Tr}\left(\mathbf{T}_{b} \mathbf{W}_{k}\right) \leq \mathrm{P}_{b, \max }, \ \forall b \in \mathcal{B},\tag{\ref{X1}{a}} \label{X1a}\\
& \ \  0 \leq \theta_{m, n} \leq 2 \pi, \ \forall m \in \mathcal{M}, \tag{\ref{X1}{b}}\label{X1b}\\
& \ \  c_{m, k} \in\{0,1\}, \ \forall m \in \mathcal{M}, \ k \in \mathcal{K}, \tag{\ref{X1}{c}}\label{X1c}\\
& \ \ {U_{\mathrm{match}} \geq \sum_{k \in \mathcal{K}} c_{m, k} , \ \forall m \in \mathcal{M} }, \tag{\ref{X1}{d}} \label{X1d}\\
& \ \  { R_{\mathrm{match}} \geq \sum_{m \in \mathcal{M}} c_{m, k} , \ \forall k \in \mathcal{K},}   \tag{\ref{X1}{e}} \label{X1e}
\end{align}
where (\ref{X1a})  denotes AP power constraints, and (\ref{X1b}) denotes modulus constraints for the RIS. The variable $\omega_{k}$ represents the weighted component corresponding to the $k$-th UE.  {Constraint (\ref{X1d}) indicates that one RIS can only serve $U_{\mathrm{match}}$ UEs at most. Constraint (\ref{X1e}) indicates that one UE can be served by at most $R_{\mathrm{match}}$ RISs.} Clearly, the optimization problem is inherently non-convex and highly complex due to the interdependency among multiple optimization variables. To address this challenge, we propose a structured low-complexity algorithmic framework designed to efficiently solve this problem.

\section{The proposed joint RIS-UE association and beamforming framework} \label{secthree}
This section presents a framework for joint RIS-UE association and beamforming to address the WSR maximization problem in (\ref{X1}). Section \ref{seca} provides an overview of the proposed two-stage framework design. In the first stage, outlined in Section \ref{sec:opt_phi_and_c}, we establish the RIS-UE association using a many-to-many matching algorithm. Section \ref{secc} details the RIS phase shift optimization, which is conducted in both the first and second stages. Finally, in the second stage, AP beamforming is optimized via the joint BD scheme, as presented in Section \ref{sec:opt_F}.
\subsection{Overview of the Proposed Two-Stage Framework} \label{seca}

To support the proposed framework design, we assume that the CSI of the RIS-assisted CF MIMO system is obtainable by the CPU through the fronthaul links \cite{CSI2}.

An intuitive illustration of the proposed two-stage framework is provided in Fig. \ref{association}. In particular, we leverage the characteristic of limited user mobility over a large timescale. In the first stage, all channels are fully acquired at the start of a large timescale to establish the RIS-UE association. Subsequently, in the second stage, only channels between each RIS and its associated UEs are acquired over the following small timescales, with RIS phase shifts and AP beamforming optimized based on partial CSI. This process is repeated for each large timescale, effectively reducing the need for frequent channel acquisitions between the RIS and non-associated UEs. As a result, this approach requires acquiring only up to $R_{\mathrm{match}}\times K$ or $M\times U_{\mathrm{match}}$ RIS-UE channels rather than the full $M\times K$ RIS-UE channels within each small timescale.

\begin{figure}[t]
    \centering   \includegraphics[width=0.7\columnwidth,height=0.3\linewidth]{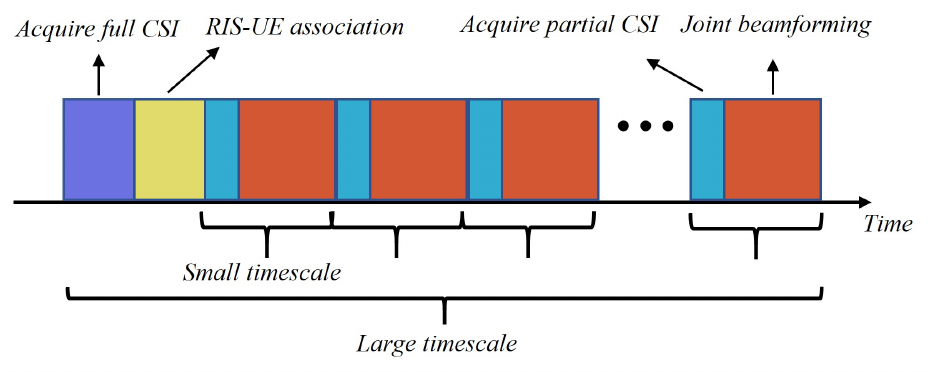}
    \caption{{Transmission protocol for the proposed RIS-assisted CF network}}
    \label{association}
\end{figure}

\subsection{ Many-to-Many Matching} 
\label{sec:opt_phi_and_c}
Since RIS-UE association is established prior to data transmission, joint beamforming optimization is not considered at this stage. Furthermore, we posit that RIS-UE association depends not only on the  channel conditions but also on the potential impact of optimizing RIS phase shifts. Building on the heuristic insight from Section \ref{sec:opt_F}, where the inter-user interference covariance matrix $\mathbf{J}_{k}$ is set to $\mathbf{0}$ under the joint BD precoding scheme, the original objective in problem (\ref{X1}) can be reformulated as 
\begin{align}
    \log \left|\mathbf{I}+\frac{1}{\sigma^{2}} \bar{\mathbf{H}}_{k} \mathbf{F}_{k} \mathbf{F}_{k}^{\mathrm{\mathit{H} }} \bar{\mathbf{H}}_{k}^{\mathrm{\mathit{H} }}\right|   =\sum_{i=1}^{N_{s}} \log \left(1+\frac{1}{\sigma^{2}} \lambda_{i}[k]^{2}\right),  \label{SVD1}
\end{align}
where the equation is obtained via the SVD of the matrix $\bar{\mathbf{H}}_{k} \mathbf{F}_{k}$ with ${\left\{\lambda_{i}[k],i=1,\dots,N_{s} \right\}}$ denoting its singular values. An upper bound can be derived as follows:
\begin{equation}
    \sum_{i=1}^{N_{s}} \log \left(1+\frac{1}{\sigma^{2}} \lambda_{i}[k]^{2}\right) \leq N_{s} \log  \left(1+\frac{1}{N_{s}\sigma^{2}} \sum_{i=1}^{N_{s}} \lambda_{i}[k]^{2}\right).
\end{equation}
{The inequality is derived using the Jensen's inequality\cite{zhanshen}, which is applicable to concave functions $f(x)$ satisfying $\sum_{i}f \left ( x_i \right ) \le I f\left( \frac{1}{I} \sum_{i}x_i \right )$.}

Note that $\sum_{i=1}^{N_{s}} \lambda_{i}[k]^{2}=\left\|\bar{\mathbf{H}}_{k} \mathbf{F}_{k}\right\|_{F}^{2}$. Using matrix norm inequalities, we obtain $\left\|\bar{\mathbf{H}}_{k} \mathbf{F}_{k}\right\|_{F}^{2} \leq\left\|\bar{\mathbf{H}}_{k}\right\|_{F}^{2}\left\|\mathbf{F}_{k}\right\|_{F}^{2}=l_{k}\left\|\bar{\mathbf{H}}_{k}\right\|_{F}^{2}$, where $l_{k}$ denotes the total transmit power from the APs to UE $k$. 
Considering that the bound of achievable rate obtained via Jensen's inequality is typically demonstrated tight in MIMO networks\cite{bound1,bound2} and recall that $\bar{\mathbf{H}}_{k}=\mathbf{H}_{k}^{d}+\sum_{m=1}^{M} c_{m,k}\mathbf{H}_{m, k}^{r} \boldsymbol{\Phi}_{m} \mathbf{G}_{m}^{r}$, the optimization problem corresponding to RIS-UE association and RIS phase shifts can thus be formulated as follows:
\begin{align}
    \max _{\boldsymbol{\Phi}, \mathbf{c}} & \ \  \sum_{k=1}^{K}\ 
 \omega_{k} \left\|\mathbf{H}_{k}^{d}+\sum_{m=1}^{M} c_{m, k} \mathbf{H}_{m, k}^{r} \boldsymbol{\Phi}_{m} \mathbf{G}_{m}^{r}\right\|_{F}^{2}\label{opt6} \\
\text { s.t. } & \ \ 0 \leq \theta_{m, n} \leq 2 \pi, \ \forall m \in \mathcal{M}, \tag{\ref{opt6}{a}}\\
&  \ \ c_{m, k} \in\{0,1\}, \ \forall m \in \mathcal{M}, \ k \in \mathcal{K}, \tag{\ref{opt6}{b}}\\
& \ \ {U_{\mathrm{match}} \geq \sum_{k \in \mathcal{K}} c_{m, k} , \ \forall m \in \mathcal{M},} \tag{\ref{opt6}{c}} \\
& \ \ { R_{\mathrm{match}} \geq \sum_{m \in \mathcal{M}} c_{m, k} , \ \forall k \in \mathcal{K},}\tag{\ref{opt6}{d}}
\end{align}
{The problem in (\ref{opt6}) can be interpreted as a \textit{many-to-many matching game}\cite{matchgame,matchnew}. In the game, the RIS set $\mathcal{M}$ and the UE set $\mathcal{K}$ are the two teams of players with $\mathcal{M} \cap \mathcal{K}  = \emptyset$.
Each player in any team tends to choose several favorite players from the opposing team, not exceeding a quota, i.e., $R_{\mathrm{match}}$ or $U_{\mathrm{match}}$. And the matching process requires the use of swap matching\cite{ swapmatch}.
Subsequently, we introduce the formal definitions of many-to-many matching.}

 \begin{definition}
 A many-to-many matching $\gamma $ is a function from the set $\mathcal{M}\cup\mathcal{K} $ to a set of all subsets of $\mathcal{M}\cup\mathcal{K} $, which satisfies the following conditions:
   \begin{itemize}
       \item[1)] {$U_{\mathrm{match}}\ge|\gamma ( m)|$ and $\gamma ( m)\subseteq \mathcal{K}$, $\forall m\in \mathcal{M} $;}
       \item[2)] {$R_{\mathrm{match}}\ge|\gamma ( k)|$ and $\gamma ( k)\subseteq \mathcal{M}$, $\forall k\in \mathcal{K} $;}
       \item[3)] {$k\in\gamma ( m)$ if and only if $m \in\gamma ( k), \forall m\in \mathcal{M}, \forall k\in \mathcal{K}$.}
   \end{itemize}
 \end{definition}
 
{From conditions 1) and 2), it is known that each RIS matches at most $U_{match}$ UEs and each UE matches at most $R_{match}$ RISs. Depending on the input parameters, the function $\gamma$ maps to different spaces. If the input variable is $m$, which represents the index of RIS, then $\gamma ( m)$ maps to the set $\mathcal{K}$. If the input variable is $k$,  which represents the index of UE, then $\gamma ( k)$ maps to the set $\mathcal{M}$. From condition 3), we have $c_{m,k}=1$ if $k\in\gamma ( m)$ and $m\in\gamma ( k)$, which represents $(m,k)$ as a mathced pair, otherwise $c_{m,k}=0$.}

 %According to the concept of two-sided exchange-stable matchings\cite{stable-match}, we can define swap-matching as follows. 
\begin{definition}
  {Each player has a proprietary preference value for each player in the opposing team, with the preference relationship denoted by the symbol $\succ$. $m\succ_{k} m' $ indicates that UE $k$ prefers RIS $m$ to serve itself rather than RIS $m'$ and  $k\succ_{m} k' $ implies that RIS $m$ prefers to serve UE $k$ rather than UE $k'$, where $m\ne m'$ and $k\ne k'$. Preferences are stored in a preference list, and the elements in the preference list of $k\in \mathcal{K} $ or $ m\in \mathcal{M} $ are sorted in descending order containing all the players of the opposing team.}
\end{definition}

{We define the UE's and RIS's preference lists as $p_k$ and $p_m$. In particular, The utility function (preference) of UE $k$ with regard to the RIS $m$ is provided  by}
\begin{equation}
 \label{Ukm}  { U_{km}=\left\|  \mathbf{H}_{m, k}^{r}  \boldsymbol{\Phi}_{m}^{k\star } \mathbf{G}_{m}^{r}\right\|_{F}^{2} ,}
\end{equation}
where
\begin{equation}
\label{phistar}  {  \boldsymbol{\Phi}_{m}^{k\star } = \mathop{\arg\max}\limits_{\boldsymbol{\Phi}_{m}}  \left\|\mathbf{H}_{k}^{d}+  \mathbf{H}_{m, k}^{r} \boldsymbol{\Phi}_{m} \mathbf{G}_{m}^{r}\right\|_{F}^{2},} 
\end{equation}
{and $\boldsymbol{\Phi}_{m}^{k\star }$ is obtained by solving \textit{Algorithm 2 (lines 2-9)}, the details of which will be presented in the next subsection. The utility function $U_{km}$ is part of the objective function in the problem (\ref{opt6}). It shows the RIS-UE channel strength when the equivalent channel strength is maximum. By the above definition, in the preference list $p_k$ of UE $k$, we have $m\succ_{k} m'$ when $U_{km} > U_{km'}$. The utility function (preference) of RIS $m$ about the UE $k$ is of the same form as in (\ref{Ukm}), then we have $k\succ_{m} k'$ when $U_{km} > U_{k'm}$ in the preference list $p_m$.}

If no player benefits from a new swap-matching, the matching $\gamma$ is called \textit{'stable'}. Then we introduce the concept of the blocking pair.
\begin{definition}
    {For a given matching function $\gamma$, and given matched pairs $(m,k')$ and $(m',k)$, $(m,k)$ is a blocking pair if $m\notin\gamma(k)$, $m\succ_{k} m'$, and $k\notin\gamma(m)$, $k\succ_{m} k'$.}
\end{definition}

{The blocking pair $(m,k)$ can only be formed if UE $k$ and RIS $m$ have not yet matched each other, which indicates the existing matches are suboptimal. In addition, we define $U_k^{\mathrm{req} } = \delta  \left\|\mathbf{H}_{k}^{d}\right\|_{F}^{2}$, where $\delta  \in(0,1)$. When $U_{km} < U_k^{\mathrm{req} }$, RIS $m$ is added to the rejection list of UE $k$, $\mathcal{R}_k$. This is because even if the RIS and the UE are successfully matched, the performance gain is insignificant when the AP-UE channel strength is considerable. Consequently, it would be more advantageous to reduce the channel acquisition overhead.}

According to the above definitions, we explain how to establish RIS-UE association in \textit{Algorithm 1}.

\begin{algorithm}
	%\textsl{}\setstretch{1.8}
	\renewcommand{\algorithmicrequire}{\textbf{Input:}}
	\renewcommand{\algorithmicensure}{\textbf{Output:}}
	\caption{{Many-to-Many Matching Algorithm}}
	\label{alg3}
	\begin{algorithmic}[1]
        \STATE {\textbf{Input:}: The utility function $U_{km}$, the preference lists $p_k$ and $p_m$, $\forall k \in \mathcal{K}, \forall m \in \mathcal{M}$.}
        \STATE {\textbf{Initialization}: $\gamma(m) =\left \{ \emptyset  \right \} $, $\gamma(k) =\left \{ \emptyset  \right \}$, $\mathcal{R} _{k} =\left \{ \emptyset  \right \}$, $\forall k \in \mathcal{K}, \forall m \in \mathcal{M}$ .}
        \STATE {\textbf{while} $\exists  p_k \ne \emptyset ,\forall k$, \textbf{do}}
        \STATE {\quad Each UE $k$ selects the first RIS $m$ in the preference list $p_k$, where $m \notin \mathcal{R}_k$.}
        \STATE {\quad $p_k\setminus m  \to  p_k$.}
        \STATE {\quad \textbf{if} $U_{km} \ge U_k^{\mathrm{req} }$ \textbf{then}}
        \STATE {\quad  \quad $\mathcal{R}_k \cup m \to \mathcal{R}_k$ and \textbf{skip current loop}        }
        \STATE{ \quad \textbf{if} $\left | \gamma(m)  \right | < U_{match} $  \textbf{then} }
        \STATE {\quad  \quad \textbf{if} $\left | \gamma(k)  \right | < R_{match}$ \textbf{then}}
        \STATE {\quad  \quad \quad  \quad $\gamma(m) \cup k\to \gamma(m)$, $\gamma(k) \cup m\to \gamma(k)$}
        \STATE {\quad  \quad \textbf{else}}
        \STATE {\quad  \quad \quad Find $m' =\mathop{\arg\min}\limits_{m''} U_{km''}, \forall m''\in  \gamma(k)$}
        \STATE {\quad  \quad \quad \textbf{if} $m\succ_{k} m'$ \textbf{then}}
        \STATE {\quad  \quad \quad \quad   $\gamma(m) \cup k\to  \gamma(m)$,   $\gamma(m')\setminus k\to  \gamma(m')$}
         \STATE {\quad  \quad \quad \quad  $\left ( \gamma(k) \setminus \left \{ m' \right \}    \right ) \cup \left \{ m \right \}  \to  \gamma(k)$}
      %  \STATE \quad  \quad \quad \textbf{else}
      %  \STATE \quad  \quad \quad \quad $\mathcal{R}_{k} \cup m \to \mathcal{R}_{k}$
        \STATE {\quad  \textbf{else} }
        \STATE {\quad  \quad Find $k' =\mathop{\arg\min}\limits_{k''} U_{k''m}, \forall k''\in \gamma(m)$}
        \STATE {\quad  \quad \textbf{if} $\left | \gamma(k)  \right | < R_{match}$ and 
        $k\succ_{m} k'$ \textbf{then} }
        \STATE {\quad  \quad \quad   $\gamma(k) \cup m\to \gamma(k)$, $\gamma(k') \setminus m\to \gamma(k')$}
     \STATE {\quad  \quad \quad   $\left ( \gamma(m) \setminus \left \{ k' \right \}    \right ) \cup \left \{ k \right \}  \to \gamma(m)$}
               {\STATE \quad  \quad \textbf{else} }
        {\STATE \quad  \quad \quad Find $m' =\mathop{\arg\min}\limits_{m''} U_{km''}, \forall m''\in  \gamma(k)$}
        {\STATE \quad  \quad \quad \textbf{if} $m\succ_{k} m'$ and $k\succ_{m} k'$ \textbf{then} }
       { \STATE \quad  \quad \quad \quad $(m,k)$ is a blocking pair and update the current match.}
       { \STATE \textbf{end while}}
		\ENSURE {  $\gamma(m)$, $\gamma(k)$, $c_{m,k}, \forall m \in \mathcal{M}, \forall k \in \mathcal{K}$}
	\end{algorithmic}  
\end{algorithm}

\subsection{RIS Phase Shift Design} \label{secc}

{In this subsection, we solve all $\boldsymbol{\Phi}_{m}^{k\star }, \forall k,m$, in (\ref{phistar}), as given in (\ref{phistar}), which are used in the first stage to initialize the utility functions, and the optimal $\boldsymbol{\Phi}_{m}, \forall m$, in (\ref{opt6}), which are utilized for data transmission in the second stage.} For the sake of narrative flow, we start by solving $\boldsymbol{\Phi}_{m}^{k\star }$  and then extend the method to solve $\boldsymbol{\Phi}_{m}$ at the end of the subsection.

{ The (\ref{phistar}) is equivalent to the following problem:}
{ \begin{align}
 \label{opt7}\max _{\boldsymbol{\Phi}_{m}^{k\star}} &
 \ \ \left\|\mathbf{H}_{k}^{d}+\mathbf{H}_{m, k}^{r} \boldsymbol{\Phi}_{m}^{k\star} \mathbf{G}_{m}^{r}\right\|_{{F }}^{2} \\
\text { s.t. } & \ \ \ 0 \leq \theta_{m, n} \leq 2 \pi,   \tag{\ref{opt7}{a}}
 \end{align}}
{for each UE $k$ and each RIS $m$.} Upon $\operatorname{vec}(\mathbf{X Y Z})=\left(\mathbf{Z}^{\mathrm{\mathit{T} }} \otimes \mathbf{X}\right) \operatorname{vec}(\mathbf{Y})$, the objective function in (\ref{opt7}) can be rewritten as
\begin{gather}  
\left\|\mathbf{H}_{k}^{d}+\mathbf{H}_{m, k}^{r} \boldsymbol{\Phi}_{m}^{k\star} \mathbf{G}_{m}^{r}\right\|_{F}^{2} \quad  \quad \quad \quad \quad \quad \quad \quad \quad \quad \quad \quad \quad  \quad \ \ \notag\\
=\operatorname{vec}^{\mathrm{\mathit{H} }}\left(\mathbf{H}_{k}^{d}+\mathbf{H}_{m, k}^{r} \boldsymbol{\Phi}_{m}^{k\star} \mathbf{G}_{m}^{r}\right) \operatorname{vec}\left(\mathbf{H}_{k}^{d}+\mathbf{H}_{m, k}^{r} \boldsymbol{\Phi}_{m}^{k\star} \mathbf{G}_{m}^{r}\right) \ \ \  \notag\\
=\operatorname{vec}^{\mathrm{\mathit{H} }}\left(\boldsymbol{\Phi}_{m}^{k\star}\right)\left(\mathbf{G}_{m}^{r \mathrm{\mathit{T} }} \otimes \mathbf{H}_{m, k}^{r}\right)^{\mathrm{\mathit{H} }}\left(\mathbf{G}_{m}^{r \mathrm{\mathit{T} }} \otimes \mathbf{H}_{m, k}^{r}\right) \operatorname{vec}\left(\boldsymbol{\Phi}_{m}^{k\star}\right) \notag\\
+\operatorname{vec}^{\mathrm{\mathit{H} }}\left(\boldsymbol{\Phi}_{m}^{k\star}\right)\left(\mathbf{G}_{m}^{r \mathrm{\mathit{T} }} \otimes \mathbf{H}_{m, k}^{r}\right)^{\mathrm{\mathit{H} }} \operatorname{vec}\left(\mathbf{H}_{k}^{d}\right) \qquad   \quad  \quad \ \ \  \notag \\
+\operatorname{vec}^{\mathrm{\mathit{H} }}\left(\mathbf{H}_{k}^{d}\right)\left(\mathbf{G}_{m}^{r  \mathrm{\mathit{T} }} \otimes \mathbf{H}_{m, k}^{r}\right) \operatorname{vec}\left(\boldsymbol{\Phi}_{m}^{k\star}\right) \qquad   \qquad  \quad \ \;  \notag \\
+\operatorname{vec}^{\mathrm{\mathit{H} }}\left(\mathbf{H}_{k}^{d}\right) \operatorname{vec}\left(\mathbf{H}_{k}^{d}\right) .   \qquad   \qquad  \qquad \quad  \; \,   \qquad   \quad  \, \,
\end{gather}
In order to  simplify the equation, we define $\mathbf{h}_{k}^{d}=\operatorname{vec}\left(\mathbf{H}_{k}^{d}\right)$, $\mathbf{C}_{mk}=\left(\mathbf{G}_{m}^{r \mathrm{\mathit{T} }} \otimes \mathbf{H}_{m, k}^{r}\right)_{\left[:,(i-1)N+i\right]}, i=1, \ldots,  N$. And $\boldsymbol{\varphi}_{m}^{k}$ is the column vector consisting of the non-zero elements of $\operatorname{vec}\left(\boldsymbol{\Phi}_{m}^{k\star}\right)$ arranged in sequence. {By removing irrelevant constant terms that are not related to the RIS phase shift matrixes, the problem (\ref{opt7}) is reformulated as }
{\begin{align}
&\label{opt8} \max _{\boldsymbol{\varphi}_{m}^{k}} \ \ \boldsymbol{\varphi}_{m}^{k\mathrm{\mathit{H} }} \mathbf{C}_{mk}^{\mathrm{\mathit{H} }} \mathbf{C}_{mk} \boldsymbol{\varphi}_{m}^{k\mathrm{\mathit{H} }}+2*\operatorname{Re}(\boldsymbol{\varphi}_{m}^{k\mathrm{\mathit{H} }} \mathbf{C}_{mk}^{\mathrm{\mathit{H} }} \mathbf{h}_{k}^{d})  \\
&\text { s.t. } \quad 0 \leq \theta_{m, n} \leq 2 \pi,  \tag{\ref{opt8}{a}}
\end{align}}
{for each UE $k$ and each RIS $m$.}
We note that (\ref{opt8}) clearly a non convex quadratic constrained quadratic
programming (QCQP) problem due to the unit modulus constraints. 

Specifically, the MM framework\cite{MM} is exploited to solve the subproblem in (\ref{opt8}). The key idea of the framework is to approximate the original issue well by iteratively maximizing the sequence of surrogate functions under the same constraints. 

Consider a maximization problem such as
\begin{align}
    &\max _{\mathbf{x}} \ \ g(\mathbf{x}) \\
&\text { s.t. }  \ \ \mathbf{x} \in \mathcal{X}. \notag
\end{align}
Its feasible points are $\left\{\mathbf{x}_{(t)}\right\} \in \mathcal{X}$. By maximizing a series of surrogate functions, i.e., $g\left(\mathbf{x} \mid \mathbf{x}_{(t)}\right), t=0,1, \ldots$, the MM approach maximizes $g\left(\mathbf{x} \right)$ while satisfying the following three requirements:
\begin{itemize}
\item[1)] $g\left(\mathbf{x}_{(t)} \mid \mathbf{x}_{(t)}\right)=g\left(\mathbf{x}_{(t)}\right)$,
\item[2)] $g\left(\mathbf{x} \mid \mathbf{x}_{(t)}\right) \leq g(\mathbf{x})$, 
\item[3)] $\mathbf{x}_{(t+1)} \in \operatorname{argmax}_{\mathbf{x} \in \mathcal{X}} \mathrm{g}\left(\mathbf{x} \mid \mathbf{x}_{(t)}\right)$.
\end{itemize}
We can readily prove that
\begin{equation}
    g\left(\mathbf{x}_{(t+1)}\right) \geq g\left(\mathbf{x}_{(t+1)} \mid \mathbf{x}_{(t)}\right) \geq g\left(\mathbf{x}_{(t)} \mid \mathbf{x}_{(t)}\right)=g\left(\mathbf{x}_{(t)}\right).
\end{equation}

Consequently, the sequence of solutions obtained at each iteration will lead to the objective function value that increases monotonically. Moreover, we know that there is an upper bound on the objective value. Hence, it can be concluded that the MM framework guarantees the convergence of the objective function to a stationary point. Subsequently, we introduce a lemma and employ it to solve the problem in (\ref{opt8}).
\begin{lemma}
For any solution $\mathbf{x}_{(t)}$ and for any feasible  $\mathbf{x}$, it holds that
\begin{align}
   \label{lemma3}    \mathbf{x}^{\mathrm{\mathit{H} }} \mathbf{L} \mathbf{x} \geq(\mathbf{x}^{\mathrm{\mathit{H} }}-\mathbf{x}_{(t)}^{\mathrm{\mathit{H} }}) \mathbf{L} \mathbf{x} 
       +\mathbf{x}^{\mathrm{\mathit{H} }} \mathbf{L}(\mathbf{x}-\mathbf{x}_{(t)})+\mathbf{x}_{(t)}^{\mathrm{\mathit{H} }} \mathbf{L} \mathbf{x}_{(t)} ,
\end{align}
where $\mathbf{L} $ is a semi-positive definite matrix. 
\end{lemma}

\begin{proof}
    Please see the appendix.
\end{proof}

It is obvious that $\mathbf{C}_{mk}^{\mathrm{\mathit{H} }} \mathbf{C}_{mk}\succeq 0$ holds, thus we have 
\begin{align}
    \boldsymbol{\varphi}_{m}^{k\mathrm{\mathit{H} }} \mathbf{C}_{mk}^{\mathrm{\mathit{H} }} \mathbf{C}_{mk} \boldsymbol{\varphi}_{m}^{k} \geq\left(\boldsymbol{\varphi}_{m}^{k\mathrm{\mathit{H} }}-\boldsymbol{\varphi}_{m,(t)}^{k\mathrm{\mathit{H} }}\right) \mathbf{C}_m{k}^{\mathrm{\mathit{H} }} \mathbf{C}_{mk} \boldsymbol{\varphi}_{m}^{k}+  \notag \\
\boldsymbol{\varphi}_{m}^{k\mathrm{\mathit{H} }} \mathbf{C}_{mk}^{\mathrm{\mathit{H} }} \mathbf{C}_{mk}\left(\boldsymbol{\varphi}_{m}^{k}-\boldsymbol{\varphi}_{m,(t)}^{k}\right)+\boldsymbol{\varphi}_{m,(t)}^{k\mathrm{\mathit{H} }} \mathbf{C}_{mk}^{\mathrm{\mathit{H} }} \mathbf{C}_{mk} \boldsymbol{\varphi}_{m,(t)}^{k}.  \ \ \  
\end{align}
{Substituting the right-hand side of the inequality into problem (\ref{opt8}) and ignore the  irrelevant constants, problem (\ref{opt8}) can be re-expressed as follows:}
{\begin{align}
 &   \max _{\boldsymbol{\varphi}_{m}^{k}} \ \ \operatorname{Re}\left\{\boldsymbol{\varphi}_{m}^{k\mathrm{\mathit{H} }}\left(\mathbf{C}_{mk}^{\mathrm{\mathit{H} }} \mathbf{C}_{mk} \boldsymbol{\varphi}_{m,(t)}^{k}+\mathbf{C}_{mk}^{\mathrm{\mathit{H} }} \mathbf{h}_{k}^{d}\right)\right\} \label{opt9}\\
&\text { s.t. } \ \ \quad0 \leq \theta_{m, n} \leq 2 \pi,  \tag{\ref{opt9}{a}}
\end{align}}
{for each UE $k$ and each RIS $m$.} The solution in $t+1$ iteration is as follows:
\begin{equation}
 {   \boldsymbol{\varphi}_{m,(t+1)}^{k}=e^{j \arg \left(\mathbf{C}_{mk}^{\mathrm{\mathit{H} }} \mathbf{C}_{mk} \boldsymbol{\varphi}_{m,(t)}^{k}+\mathbf{C}_{mk}^{\mathrm{\mathit{H} }} \mathbf{h}_{k}^{d}\right)}} .\label{phik}
\end{equation}
In this way, we have $\boldsymbol{\Phi}_{m}^{k\star } = \mathrm{diag} (\boldsymbol{\varphi}_{m}^{k\star} )$, where $\boldsymbol{\varphi}_{m}^{k\star}$ is the $\boldsymbol{\varphi}_{m}^{k}$ obtained in the last iteration. {Then we apply $\boldsymbol{\Phi}_{m}^{k\star }$ to the initialization of the utility function in (\ref{Ukm}) to establish the RIS-UE association in the first stage. The specific optimization details can be seen in \textit{Algorithm 2 (lines 2-9)}.}

{With the RIS-UE association fully established, the matching-related variables and constraints in problem (\ref{opt6}) can be omitted, thereby allowing  (\ref{opt6}) to be reformulated in the same structure as (\ref{opt9}), representing the RIS phase shifts optimization for data transmission in the second stage, as follows:}
\begin{align}
 &   \max _{\boldsymbol{\varphi}_{m}} \ \ \operatorname{Re}\left\{\boldsymbol{\varphi}_{m}^{\mathrm{\mathit{H} }}\left(\mathbf{C}_{m} \boldsymbol{\varphi}_{m,(t)}+ \mathbf{h}_{m}\right)\right\} \label{optnew}\\
&\text { s.t. } \ \ \quad0 \leq \theta_{m, n} \leq 2 \pi,  \tag{\ref{optnew}{a}}
\end{align}
where $\mathbf{C}_{m}=\sum_{k\in\gamma(m)} \mathbf{C}_{mk}^{\mathrm{\mathit{H} }}\mathbf{C}_{mk}$ and $\mathbf{h}_{m}=\sum_{k\in\gamma(m)} \mathbf{C}_{mk}^{\mathrm{\mathit{H} }}\mathbf{h}_{k}^{d}$. And the solution in $t+1$th iteration has the same form as in (\ref{phik}), which can be expressed as
\begin{equation}
 {    \boldsymbol{\varphi}_{m,(t+1)}=e^{j \arg \left(\mathbf{C}_{m} \boldsymbol{\varphi}_{m,(t)}+ \mathbf{h}_{m}\right)}} .\label{phim}
\end{equation}
{We summarize the procedure in \textit{Algorithm 2 (lines 10-17)}.}

\begin{algorithm}
	%\textsl{}\setstretch{1.8}
	\renewcommand{\algorithmicrequire}{\textbf{Input:}}
	\renewcommand{\algorithmicensure}{\textbf{Output:}}
	\caption{{MM Algorithm }}
	\label{alg2}
	\begin{algorithmic}[1]
       \STATE \textbf{Initialization}: { $t=1$, the accuracy $\epsilon $, all initial RIS phase shift vectors $\boldsymbol\varphi_{m,(\mathrm{1})}^{k}  $ and $\boldsymbol\varphi_{m,(\mathrm{1})}  $, and all $\mathbf{C}_{mk}$, $\mathbf{C}_{m}$, $\mathbf{h}_{k}^{d}$, $\mathbf{h}_{m}$, $t_{\mathrm{iter} }$}
        \STATE { \textbf{If} \textit{the algorithm is used to solve (\ref{phistar})} \textbf{then}}
        \STATE { \quad Define the value of the objective in problem in (\ref{opt9}) \\ as $g(\boldsymbol\varphi_{m,(t)}^{k})$ and calculate  $g(\boldsymbol\varphi_{m,(\mathrm{1})}^{k})$.}
        \STATE { \quad \textbf{while}  $|g(\boldsymbol\varphi_{m,(t)}^{k})-g(\boldsymbol\varphi_{m,(\mathrm{t-1})}^{k})| > \epsilon$ or $t<t_{\mathrm{iter} }$ \textbf{do}}
        \STATE { \quad \quad Update $\boldsymbol\varphi_{m,(\mathrm{t+1})}^{k}$ based on  (\ref{phik}) .}
        \STATE { \quad \quad  Calculate $g(\boldsymbol\varphi_{m,(\mathrm{t+1})}^{k})$.}
        \STATE { \quad \quad  Update $t=t+1$.}
        \STATE { \quad \textbf{end while}}
        \STATE  { \textbf{Output:} all $\boldsymbol{\Phi}_{m}^{k\star } = \mathrm{diag} (\boldsymbol{\varphi}_{m}^{k\star} )$}
        \STATE \textbf{If} \textit{the algorithm is used to solve problem in
 (\ref{opt6}) when the matching variables are known} \textbf{then}
        \STATE \quad Define the value of the objective in problem in (\ref{optnew}) \\ as $g(\boldsymbol\varphi_{m,(t)})$ and calculate  $g(\boldsymbol\varphi_{m,(\mathrm{1})})$.
        \STATE \quad \textbf{while}  $|g(\boldsymbol\varphi_{m,(t)})-g(\boldsymbol\varphi_{m,(\mathrm{t-1})})| > \epsilon$ or $t<t_{\mathrm{iter} }$ \textbf{do}
        \STATE \quad \quad Update $\boldsymbol\varphi_{m,(\mathrm{t+1})}$ based on (\ref{phim}) .
        \STATE \quad \quad  Calculate $g(\boldsymbol\varphi_{m,(\mathrm{t+1})}$.
        \STATE \quad \quad  Update $t=t+1$.
        \STATE \quad \textbf{end while}
        \STATE  \textbf{Output:} all $\boldsymbol{\Phi}_{m}^{\star } = \mathrm{diag} (\boldsymbol{\varphi}_{m}^{\star} )$
		
	\end{algorithmic}  
\end{algorithm}

\subsection{AP Beamforming Design} 
\label{sec:opt_F}

Given the optimal $\mathbf{c}$ and $\boldsymbol{\Phi}$, the  problem reduces to one involving only $\mathbf{F}$. Specifically, by removing extraneous variables and constraints, the problem in (\ref{X1}) can be simplified and reformulated as  follows:
\begin{align}
\max _{\mathbf{F}_{k}, \forall k \in \mathcal{K}} & \ \ \sum_{k = 1}^{K} \omega_{k} \log \left|\mathbf{I}+\bar{\mathbf{H}}_{k} \mathbf{F}_{k} \mathbf{F}_{k}^{\mathrm{\mathit{H} }} \bar{\mathbf{H}}_{k}^{\mathrm{\mathit{H} }} (\mathbf{J}_{k} + \sigma^{2} \mathbf{I})^{-1}\right| \label{X2}\\
\text { s.t. } & \ \ \sum_{k = 1}^{K} \operatorname{Tr}\left(\mathbf{T}_{b} \mathbf{W}_{k}\right) \leq \mathrm{P}_{b, \max }, \ \forall b \in \mathcal{B}. \tag{\ref{X2}{a}}
\end{align}

To effectively manage inter-user interference, we employ a block diagonalization (BD) precoding scheme to eliminate interference among users \cite{5594707}. Zero-forcing (ZF) constraints are introduced for each UE $k$, i.e., 
\begin{equation}
    \bar{\mathbf{H}}_{i} \mathbf{F}_{k}=\mathbf{0}, \ \forall k \neq i ,
\end{equation}
which are equivalent to
\begin{equation}
\bar{\mathbf{H}}_{i} \mathbf{W}_{k} \bar{\mathbf{H}}_{i}^{\mathrm{\mathit{H} }}=\mathbf{0}, \ \forall k \neq i .   
\end{equation}

With the introduction of ZF constraints, $\mathbf{J}_{k}=\mathbf{0}$, allowing the optimization problem to be further reformulated as 
\begin{align}
\label{opt1}\max _{\mathbf{W}_{k}, \forall k \in \mathcal{K}} & 
 \ \ \sum_{k=1}^{K} \omega_{k} \log \left|\mathbf{I}+\frac{1}{\sigma^{2}} \bar{\mathbf{H}}_{k} \mathbf{W}_{k} \bar{\mathbf{H}}_{k}^{\mathit{H} }\right|  \\
\text { s.t. } &  \ \ \sum_{k=1}^{K} \operatorname{Tr}\left(\mathbf{T}_{b} \mathbf{W}_{k}\right) \leq \mathrm{P}_{b, \max }, \ \forall b \in \mathcal{B},\tag{\ref{opt1}{a}} \\
&  \ \ \bar{\mathbf{H}}_{i} \mathbf{W}_{k} \bar{\mathbf{H}}_{i}^{\mathrm{\mathit{H} }}=\mathbf{0}, \ \forall k \neq i.\tag{\ref{opt1}{b}}
\end{align}
It is straightforward to verify that the problem above is a convex optimization problem, solvable with standard optimization tools such as CVX \cite{cvx}. However, this approach often incurs substantial computational overhead and lacks a closed-form optimal solution. Therefore, we employ the Lagrangian multiplier method to derive an optimal closed-form solution.

We start by removing the ZF constraints as follows: We define
\begin{equation*}
\mathbf{D}_{k}=\left[\bar{\mathbf{H}}_{1}^{\mathrm{\mathit{T} }}, \ldots, \bar{\mathbf{H}}_{k-1}^{\mathrm{\mathit{T} }}, \bar{\mathbf{H}}_{k+1}^{\mathrm{\mathit{T} }}, \ldots, \bar{\mathbf{H}}_{K}^{\mathrm{\mathit{T} }}\right]^{\mathrm{\mathit{T} }}, \ \forall k\in \mathcal{K} ,
\end{equation*}
where $\mathbf{D}_{k} \in \mathbb{C}^{L_{1} \times L_{2}}$,  with $L_{1}=(K-1)N_{r}$ and $L_{2}=BN_{t}$. $\mathbf{D}_{k}$ can be decomposed as $\mathbf{D}_{k}=\mathbf{U}_{k} \boldsymbol{\Sigma}_{k} \mathbf{E}_{k}^{\mathrm{\mathit{H} }}$ using the (full) singular value decomposition (SVD), where $\boldsymbol{\Sigma}_{k} \in \mathbb{C}^{L_{1} \times L_{2}}$ is an eigenvalue matrix.
 $\mathbf{U}_{k} \in \mathbb{C}^{L_{1} \times L_{1}}$ and $\mathbf{E}_{k} \in \mathbb{C}^{L_{2} \times L_{2}}$ are both unitary matrices.
 
 Recall that ${KN_r } \le{BN_t }$, thus $\operatorname{Rank}\left(\mathbf{D}_{k}\right)=L_{1}<L_{2}$, then we have
 \begin{equation}
     \mathbf{D}_{k}=\mathbf{U}_{k} \left[\breve{\mathbf{\Sigma}}_{k}, {\mathbf{0}}_{ L_1 \times \left ( L_2-L_1 \right )   }\right]\left[\breve{\mathbf{E}}_{k}, \tilde{\mathbf{E}}_{k}\right]^{\mathrm{\mathit{H} }},
 \end{equation}
 where $\breve{\mathbf{E}}_{k}\in \mathbb{C}^{L_{2} \times L_{1}}$ comprises the  $L_{1}$ singular vectors  associated with the non-negative singular values of $\breve{\mathbf{\Sigma}}_{k}$,  and $\tilde{\mathbf{E}}_{k}\in \mathbb{C}^{L_{2} \times (L_{2}-L_{1})}$ contains the  $(L_{2}-L_{1})$ singular vectors corresponding to the $(L_{2}-L_{1})$ zero singular values.
 They satisfy that $\mathbf{D}_{k} \tilde{\mathbf{E}}_{k}=\mathbf{0}$,  $\breve{\mathbf{E}}_{k}^{\mathrm{\mathit{H} }}\tilde{\mathbf{E}}_{k}=\mathbf{0}$, and $\tilde{\mathbf{E}}_{k}^{\mathrm{\mathit{H} }}\tilde{\mathbf{E}}_{k}=\mathbf{I}$. Furthermore, we introduce the following lemma. 
% presented in\cite{5594707}.
\begin{lemma}
The optimal structure of $\mathbf{F}_{k}$ and $\mathbf{W}_{k}$ are provided by
\begin{equation}
   \mathbf{F}_{k}=\tilde{\mathbf{E}}_{k} \bar{\mathbf{S}}_k , \  \forall k\in \mathcal{K},
\end{equation}
\begin{equation}
   \mathbf{W}_{k}=\tilde{\mathbf{E}}_{k} \mathbf{S}_{k} \tilde{\mathbf{E}}_{k}^{\mathrm{\mathit{H} }},  \ \forall k\in \mathcal{K},
\end{equation}
where $  \mathbf{S}_{k}= \bar{\mathbf{S}}_k  \bar{\mathbf{S}}_k  ^{H}  \in \mathbb{C}^{\left(L_{2}-L_{1}\right) \times\left(L_{2}-L_{1}\right)}$ and $\mathbf{S}_{k} \succeq 0$. Meanwhile, $\mathbf{S}_{k}$ is the new optimization variable.
\end{lemma}

\begin{proof}
    Without loss of generality, the precoding matrix $\mathbf{F}_{k}$ consists of two parts when using the BD scheme. The first part $\tilde{\mathbf{E}}_{k}$ is used to eliminate inter-user interference and the second part  $\bar{\mathbf{S}}_k$ is used to improve performance and for potential power control.
\end{proof}

Upon $\mathbf{D}_{k} \tilde{\mathbf{E}}_{k}=\mathbf{0}$, it is apparent  that $\bar{\mathbf{H}}_{i} \tilde{\mathbf{E}}_{k}=\mathbf{0},\forall k\ne i$. Therefore, by following the optimal structure of $\mathbf{W}_{k}$ proposed by  Lemma 2, we can remove the ZF constraints. The equivalent optimization problem is
\begin{align}
\label{opt2}    \max _{\mathbf{S}_{k}, \forall k \in \mathcal{K}} & 
\ \  \sum_{k=1}^{K} \omega_{k} \log \left|\mathbf{I}+\frac{1}{\sigma^{2}} \bar{\mathbf{H}}_{k} \tilde{\mathbf{E}}_{k} \mathbf{S}_{k} \tilde{\mathbf{E}}_{k}^{\mathit{H} } \bar{\mathbf{H}}_{k}^{\mathit{H} }\right| \\
\text { s.t. } & \ \ \sum_{k=1}^{K} \operatorname{Tr}\left(\mathbf{T}_{b} \tilde{\mathbf{E}}_{k} \mathbf{S}_{k} \tilde{\mathbf{E}}_{k}^{\mathit{H} }\right) \leq P_{b, \max }, \ \forall b \in \mathcal{B}. \tag{\ref{opt2}{a}}
\end{align}
The problem is convex as well as the problem (\ref{opt1}), and the Lagrangian function of the problem (\ref{opt2}) can be expressed as:
\begin{align}
    \mathcal{L}\left(\mathbf{S}_{k}, \mu_{b}, \forall k, b\right)=\sum_{k=1}^{K} \omega_{k} \log \left|\mathbf{I}+\frac{1}{\sigma^{2}}\bar{\mathbf{H}}_{k} \tilde{\mathbf{E}}_{k} \mathbf{S}_{k} \tilde{\mathbf{E}}_{k}^{\mathit{H} } \bar{\mathbf{H}}_{k}^{\mathit{H} }\right| \notag\\
-\sum_{b=1}^{B} \mu_{b}\left(\sum_{k=1}^{K} \operatorname{Tr}\left(\mathbf{T}_{b} \tilde{\mathbf{E}}_{k} \mathbf{S}_{k} \tilde{\mathbf{E}}_{k}^{\mathit{H} }\right)-P_{b, \max }\right),
\end{align}
where $\mu_{b} \geq 0,\forall b \in \mathcal{B}$ is the non-negative dual variable related to the $b$-th AP's power constraint. Next, we define the Lagrange dual function for (\ref{opt2}) as
\begin{equation}
  g\left(\left\{\mu_{b}\right\}\right)=\max  \mathcal{L}\left(\mathbf{S}_{k}, \mu_{b}, \forall k, b\right)\label{opt3} 
\end{equation}
Additionally, the definition of the dual problem of (\ref{opt2}) is 
%\begin{equation}
%   \min _{\mu_{b} \geq 0, \forall b \in \mathcal{B} } g\left(\left\{\mu_{b}\right\}\right)\label{opt4} 
%\end{equation}
{\begin{align}
   \min _{\mu_{b} , \forall b \in \mathcal{B}} & 
\ \  g\left(\left\{\mu_{b}\right\}\right)\label{opt4}  \\
\text { s.t. } & \ \ \mu_{b} \geq 0, \forall b \in \mathcal{B} . \tag{\ref{opt4}{a}}
\end{align}}
It  can be further known that the duality gap between the optimal solutions of (\ref{opt3}) and (\ref{opt4}) is zero by using  Slater’s condition. Therefore,  to satisfy the complementary slackness condition for the power constraints, the value of $\mu_{b}$ should be  such that
\begin{equation}
  \label{complementary}   \mu_{b}\left(\sum_{k=1}^{K} \operatorname{Tr}\left(\mathbf{T}_{b} \tilde{\mathbf{E}}_{k} \mathbf{S}_{k} \tilde{\mathbf{E}}_{k}^{\mathrm{\mathit{H} }}\right)-P_{b, \max }\right)=0, \ \forall b \in \mathcal{B} .
\end{equation}

In the following, instead of solving for the optimal $\mu_{b}$, We first detail the process of obtaining the  optimal closed-form solution of $\mathbf{S}_{k}$ with a fixed set of $\left\{\mu_{b},\forall b \in \mathcal{B}\right\}$. Note that (\ref{opt3}) can be divided into $K$ independent subproblems, the corresponding subproblem for each UE $k$ is given as
\begin{align}
  \label{opt5}  \max _{\mathbf{S}_{k} \succeq  0} \ \ \omega_{k} \log \left|\mathbf{I}+\frac{1}{\sigma^{2}} \bar{\mathbf{H}}_{k} \tilde{\mathbf{E}}_{k} \mathbf{S}_{k} \tilde{\mathbf{E}}_{k}^{\mathrm{\mathit{H} }} \bar{\mathbf{H}}_{k}^{\mathrm{\mathit{H} }}\right| 
-\operatorname{Tr}\left(\mathbf{T}_{\mu} \tilde{\mathbf{E}}_{k} \mathbf{S}_{k} \tilde{\mathbf{E}}_{k}^{\mathrm{\mathit{H} }}\right)
\end{align}
where $\mathbf{T}_{\mu}=\sum_{b=1}^{B} \mu_{b} \mathbf{T}_{b}$ is a diagonal matrix consisting of $\left\{\mu_{b},\forall b \in \mathcal{B}\right\}$.

 From (\ref{complementary}), we can observe that corresponding power constraint is tight, i.e., $\sum_{k=1}^{K} \operatorname{Tr}\left(\mathbf{T}_{b} \tilde{\mathbf{E}}_{k} \mathbf{S}_{k} \tilde{\mathbf{E}}_{k}^{\mathrm{\mathit{H} }}\right)=P_{b, \max }$, only when $\mu_{b}>0$. Furthermore, we define $\mathcal{D}=\left \{\mu_{b} | \ \mu_{b}>0 , \forall b\in \mathcal{B} \right \}$ and $D_{\mu}=\left | \mathcal{D} \right | $.  
 {To ensure that problem (\ref{opt5}) can converge to a stationary point, it holds that  $D_{\mu}=B \geq\left\lceil\frac{B N_{t}-(K-1) N_{r}}{N_{t}}\right\rceil$.}
 %the detailed proof  can be found in \cite{5594707}. 

Then, with $L_{1}=(K-1)N_{r}$ and $L_{2}=BN_{t}$, without loss of generality we assume that $D_{\mu} N_{t} \geq L_{2}-L_{1}$, since the only scenario in which we are interested is when the objective value of problem (\ref{opt5}) is bounded.
Therefore, we have $\operatorname{Rank}\left(\tilde{\mathbf{E}}_{k}^{\mathit{H} } \mathbf{T}_{\mu} \tilde{\mathbf{E}}_{k}\right)=\min \left(L_{2}-L_{1}, D_{\mu} N_{t}\right)=L_{2}-L_{1}$ and the matrix $\tilde{\mathbf{E}}_{k}^{\mathit{H} } \mathbf{T}_{\mu} \tilde{\mathbf{E}}_{k} \in \mathbb{C}^{\left(L_{2}-L_{1}\right) \times\left(L_{2}-L_{1}\right)}$ is full rank and invertible. By using $\operatorname{Tr}(\mathbf{A B})=\operatorname{Tr}(\mathbf{B} \mathbf{A})$, we have:
\begin{equation*}
      \operatorname{Tr}(\mathbf{T}_{\mu} \tilde{\mathbf{E}}_{k} \mathbf{S}_{k} \tilde{\mathbf{E}}_{k}^{\mathrm{\mathit{H} }})=\operatorname{Tr}((\tilde{\mathbf{E}}_{k}^{\mathrm{\mathit{H} }} \mathbf{T}_{\mu} \tilde{\mathbf{E}}_{k})^{1 / 2} \mathbf{S}_{k}(\tilde{\mathbf{E}}_{k}^{\mathrm{\mathit{H} }} \mathbf{T}_{\mu} \tilde{\mathbf{E}}_{k})^{1 / 2}).
\end{equation*}
We further define  $\tilde{\mathbf{S}}_{k}$ as
\begin{equation}
   \label{eq1} \tilde{\mathbf{S}}_{k}=\left(\tilde{\mathbf{E}}_{k}^{\mathit{H} } \mathbf{T}_{\mu} \tilde{\mathbf{E}}_{k}\right)^{1 / 2} \mathbf{S}_{k}\left(\tilde{\mathbf{E}}_{k}^{\mathit{H} } \mathbf{T}_{\mu} \tilde{\mathbf{E}}_{k}\right)^{1 / 2},
\end{equation}
where $\tilde{\mathbf{S}}_{k} \in \mathbb{C}^{\left(L_{2}-L_{1}\right) \times\left(L_{2}-L_{1}\right)}$. By substituting (\ref{eq1}) into (\ref{opt5}), the maximization problem is re-described as
\begin{align}
    \max _{\tilde{\mathbf{S}}_{k} \succeq 0} \ \ \omega_{k} \log \mid \mathbf{I}+\frac{1}{\sigma^{2}} \bar{\mathbf{H}}_{k} \tilde{\mathbf{E}}_{k}\left(\tilde{\mathbf{E}}_{k}^{\mathrm{\mathit{H} }} \mathbf{T}_{\mu} \tilde{\mathbf{E}}_{k}\right)^{-1 / 2} \tilde{\mathbf{S}}_{k} 
\left(\tilde{\mathbf{E}}_{k}^{\mathrm{\mathit{H} }} \mathbf{T}_{\mu} \tilde{\mathbf{E}}_{k}\right)^{-1 / 2} \tilde{\mathbf{E}}_{k}^{\mathrm{\mathit{H} }} \bar{\mathbf{H}}_{k}^{\mathrm{\mathit{H} }} \mid-\operatorname{Tr}\left(\tilde{\mathbf{S}}_{k}\right) \label{eq2}
\end{align}
where $\bar{\mathbf{H}}_{k} \tilde{\mathbf{E}}_{k}(\tilde{\mathbf{E}}_{k}^{\mathrm{\mathit{H} }} \mathbf{T}_{\mu} \tilde{\mathbf{E}}_{k})^{-1 / 2} \in \mathbb{C}^{N_{r} \times\left(L_{2}-L_{1}\right)}$. Note that its rank is $N_{r}$ since $N_{r}<L_{2}-L_{1}$, by using the  (reduced) SVD, we can get
\begin{equation}
    \bar{\mathbf{H}}_{k} \tilde{\mathbf{E}}_{k}\left(\tilde{\mathbf{E}}_{k}^{\mathrm{\mathit{H} }} \mathbf{T}_{\mu} \tilde{\mathbf{E}}_{k}\right)^{-1 / 2}=\hat{\mathbf{U}}_{k} \hat{\boldsymbol{\Sigma}}_{k} \hat{\mathbf{E}}_{k}^{\mathrm{\mathit{H} }}, \label{eq3}
\end{equation}
where $\hat{\mathbf{E}}_{k} \in \mathbb{C}^{\left(L_{2}-L_{1}\right) \times N_{r}}$, $\hat{\mathbf{E}}_{k}^{\mathrm{\mathit{H} }} \hat{\mathbf{E}}_{k}=\mathbf{I}$, and $\hat{\boldsymbol{\Sigma}}_{k}=\operatorname{diag}\left(\hat{\sigma}_{k, 1}, \ldots, \hat{\sigma}_{k, N_{r}}\right)$. Furthermore, substituting (\ref{eq3}) into (\ref{eq2}), the optimization problem becomes 
\begin{equation}
    \max _{\mathbf{\Lambda}_{k}} \ \ \omega_{k} \log \left|\mathbf{I}+\frac{1}{\sigma^{2}} \hat{\boldsymbol{\Sigma}}_{k}^{2} \boldsymbol{\Lambda}_{k}\right|-\operatorname{Tr}\left(\boldsymbol{\Lambda}_{k}\right)
\end{equation}
where $\boldsymbol{\Lambda}_{k}=\hat{\mathbf{E}}_{k}^{\mathrm{\mathit{H} }} \tilde{\mathbf{S}}_{k} \hat{\mathbf{E}}_{k}$. By applying Hadamard’s inequality and water-filling algorithm, the optimal $\boldsymbol{\Lambda}_{k} $ can be expressed as 
\begin{equation}
    \boldsymbol{\Lambda}_{k}=\operatorname{Diag}\left(\lambda_{k, 1}, \ldots \lambda_{k, N_{r}}\right),
\end{equation}
where $\lambda_{k, i}=\max \left(0, \omega_{k}-\frac{1}{\hat{\sigma}_{k, i}^{2}}\right)$.
 
To sum up, subject to $\tilde{\mathbf{S}_{k}}=\hat{\mathbf{E}}_{k} \boldsymbol{\Lambda}_{k} \hat{\mathbf{E}}_{k}^{\mathrm{\mathit{H} }}$ and Lemma 2,  the optimal solutions of (\ref{opt5}) and (\ref{opt1}) are
\begin{equation}
\mathbf{S}_{k}=\left(\tilde{\mathbf{E}}_{k}^{\mathrm{\mathit{H} }} \mathbf{T}_{\mu} \tilde{\mathbf{E}}_{k}\right)^{-\frac{1}{2}} \hat{\mathbf{E}}_{k} \boldsymbol{\Lambda}_{k} \hat{\mathbf{E}}_{k}^{\mathrm{\mathit{H} }}\left(\tilde{\mathbf{E}}_{k}^{\mathrm{\mathit{H} }} \mathbf{T}_{\mu} \tilde{\mathbf{E}}_{k}\right)^{-\frac{1}{2}} \label{Qk},
\end{equation}
\begin{align}
\mathbf{W}_{k}^{\mathrm{opt}}&=\tilde{\mathbf{E}}_{k}\left(\tilde{\mathbf{E}}_{k}^{\mathrm{\mathit{H} }} \mathbf{T}_{\mu}^{\mathrm{opt}} \tilde{\mathbf{E}}_{k}\right)^{-\frac{1}{2}} \hat{\mathbf{E}}_{k} \boldsymbol{\Lambda}_{k}
 \hat{\mathbf{E}}_{k}^{\mathrm{\mathit{H} }}\left(\tilde{\mathbf{E}}_{k}^{\mathrm{\mathit{H} }} \mathbf{T}_{\mu}^{\mathrm{opt}} \tilde{\mathbf{E}}_{k}\right)^{-\frac{1}{2}} \tilde{\mathbf{E}}_{k}^{\mathrm{\mathit{H} }} \label{Wk}.
\end{align} 
for each UE $k$, where $\mathbf{T}_{\mu}^{\mathrm{opt}}=\sum_{b=1}^{B} \mu_{b}^{\mathrm{opt}} \mathbf{T}_{b}$.

Notice that the above discussion of solving  the $\mathbf{W}_{k}^{\mathrm{opt}}$ is  under the assumption that we have obtained the optimal $\mu_{b}$. Next, we will focus on  solving  the optimal $\mu_{b}$. In (\ref{complementary}), we mentioned that the $\mu_{b}$ should satisfy the complementary condition, thus we construct a function concerning  $\mu_{b}$, which can be expressed as $f_{b}\left(\mu_{b}\right)=\sum_{k=1}^{K} \operatorname{Tr}\left(\mathbf{T}_{b} \tilde{\mathbf{E}}_{k} \mathbf{S}_{k} \tilde{\mathbf{E}}_{k}^{\mathrm{\mathit{H} }}\right)=\sum_{k=1}^{K} \operatorname{Tr}\left(\mathbf{T}_{b} \mathbf{W}_{k}\right)$. We further substitute (\ref{Wk}) into $f_{b}\left(\mu_{b}\right)$, then we have  
\begin{small}
\begin{align}
   f_{b}\left(\mu_{b}\right)
={} & \sum_{k=1}^{K} \operatorname{Tr}\left(\mathbf{T}_{b} \tilde{\mathbf{E}}_{k}\left(\tilde{\mathbf{E}}_{k}^{\mathrm{\mathit{H} }} \mathbf{T}_{\mu} \tilde{\mathbf{E}}_{k}\right)^{-\frac{1}{2}} \tilde{\mathbf{S}}_{k}\left(\tilde{\mathbf{E}}_{k}^{\mathit{H} } \mathbf{T}_{\mu} \tilde{\mathbf{E}}_{k}\right)^{-\frac{1}{2}} \tilde{\mathbf{E}}_{k}^{\mathrm{\mathit{H} }}\right) \notag\\
\label{f_mu_b}\stackrel{(a)}{=}{} & \operatorname{Tr}\left(\sum_{k=1}^{K} \mathbf{Z}_{b, k}^{1}\left(\mu_{b} \mathbf{Z}_{b, k}^{1}+\mathbf{Z}_{b, k}^{2}\right)^{-\frac{1}{2}} \tilde{\mathbf{S}}_{k}\left(\mu_{b} \mathbf{Z}_{b, k}^{1}+\mathbf{Z}_{b, k}^{2}\right)^{-\frac{1}{2}}\right) ,
\end{align}
\end{small}
where $\mathbf{Z}_{b, k}^{1}=\tilde{\mathbf{E}}_{k}^{\mathrm{\mathit{H} }} \mathbf{T}_{b} \tilde{\mathbf{E}}_{k}$ and $\mathbf{Z}_{b, k}^{2}=\sum_{i \neq b}^{B} \mu_{i} \tilde{\mathbf{E}}_{k}^{\mathrm{\mathit{H} }} \mathbf{T}_{i} \tilde{\mathbf{E}}_{k}$, and (a) is obtained by the fundamental properties of matrix trace. 
%Moreover, in (a), we separate $\mu_{b}$ from the other variables and attempt to work with it further to find out the monotonicity of $f_{b}\left(\mu_{b}\right)$. 
Observing the monotonicity of $f_{b}\left(\mu_{b}\right)$ can be challenging; therefore, we introduce the following lemma.
 %Then, we consider the derivation of $f_{b}\left(\mu_{b}\right)$ with respect to $\mu_{b}$:
 \begin{lemma}
$f_{b}\left(\mu_{b}\right)$ is a monotonically decreasing function when $\mu_{b}\ge 0$. 
\end{lemma}

\begin{proof}
    Please see the appendix.
\end{proof}

According to Lemma 3, if $f_{b}\left(0\right) \le \mathrm{P}_{b, \max}$, then 
$\mu_{b}^{\mathrm{opt}} =0$. 
%Otherwise, since $f_{b}\left(\infty \right) =0$, we can use bisection  method to find the optimal $\mu_{b}$ thus there must exist an optimal $\mu_{b}$ that satisfies the following equation:
Otherwise, since $f_{b}\left(\infty \right) =0$, the optimal $\mu_{b}$ can be found  using the bisection method, which is based on the following equation:
\begin{equation}
f_{b}\left(\mu_{b}^{\mathrm{opt}}\right)=P_{b, \max }.
\end{equation}
%Meanwhile we can use the scaling method to set up an initial upper bound of $\mu_{b}$. Following this, we can obtain the optimal $\mu_{b}, \forall b \in \mathcal{B}$. 
Then, by substituting  all $\mu_{b}^{\mathrm{opt}}$ into (\ref{Wk}), we derive the optimal solution for $\mathbf{W}_{k}^{\mathrm{opt}}, \forall k \in \mathcal{K}$.

We summarize the algorithm for solving the problem (\ref{opt1}) in Algorithm 3.

\begin{algorithm}
	%\textsl{}\setstretch{1.8}
	\renewcommand{\algorithmicrequire}{\textbf{Input:}}
	\renewcommand{\algorithmicensure}{\textbf{Output:}}
	\caption{Joint BD Algorithm to solve Problem in (\ref{opt1})}
	\label{alg1}
	\begin{algorithmic}[1]
        \STATE \textbf{Initialization}: $\mu_{b}\ge 0$, the upper bound $\mu_{b}^{up}$ and the lower bound $\mu_{b}^{lp}$, $\forall b \in \mathcal{B}$, the accuracy $\varepsilon $ and $\zeta$.
		\STATE  Given $\mu_{b}$'s , calculate initial $\mathbf{S}_{k}$ based on (\ref{Qk}).
		\REPEAT
        \STATE \textbf{for} $b=1:B$ \textbf{do}
		\STATE \quad Fixed $\mu_{i},\forall i\ne b$, using the following bisection search method to find $\mu_{b}$.
        \STATE \quad\textbf{while} $\mu_{b}^{up}-\mu_{b}^{lp}>\varepsilon $ \textbf{do}  
		\STATE \quad\quad If $f_{b}\left(0\right) \le \mathrm{P}_{b, \max}$ holds, then $\mu_{b}=0$, and jump out of the while loop; Otherwise go to the next step.
		\STATE \quad\quad Calculate $\mu_{b}=(\mu_{b}^{up}+\mu_{b}^{lp})/2$.
		\STATE \quad\quad If  $f_{b}\left(\mu_{b}\right) \ge \mathrm{P}_{b, \max}$, let $\mu_{b}^{lp}=\mu_{b}$; Otherwise, let $\mu_{b}^{up}=\mu_{b}$.
          \STATE \quad\quad Update $\mathbf{T}_{\mu}$ based on $\mathbf{T}_{\mu}=\sum_{b=1}^{B} \mu_{b} \mathbf{T}_{b}$.
        \STATE \quad\quad Update all $\mathbf{S}_{k}$'s based on (\ref{Qk}).
        \STATE \quad\textbf{end while}
        \STATE \textbf{end for}
		\UNTIL All $\mu_{b}$'s coverage and the error is less than $\zeta$ .
		\STATE   Calculate $\mathbf{W}_{k}^{\mathrm{opt}}$ based on Lemma 2.
		\ENSURE   $\mathbf{W}_{k}^{\mathrm{opt}}, \forall  k\in \mathcal{K} $
	\end{algorithmic}  
\end{algorithm}

\begin{algorithm}
	%\textsl{}\setstretch{1.8}
	\renewcommand{\algorithmicrequire}{\textbf{Input:}}
	\renewcommand{\algorithmicensure}{\textbf{Output:}}
	\caption{{The Two-Stage Framework to solve Problem in (\ref{X1})}}
	\label{alg4}
	\begin{algorithmic}[1]
        \STATE \textbf{Initialization}: $\mu_{b}$, $\mu_{b}^{up}$, $\mu_{b}^{lp}$, $\varepsilon $, $\zeta$, $t$, $\epsilon $, $\mathbf{\Phi_m } $, $\mathbf{G}_{b,m}^{r}$, $\mathbf{H}_{b,k}^{d}$, $\mathbf{H}_{m,k}^{r}$, $\forall b \in \mathcal{B}, \forall k \in \mathcal{K}, \forall m \in \mathcal{M}$. 
		\STATE { Solve  (\ref{phistar}) via MM Algorithm (lines 2-9) and output the utility function $U_{km}, \forall m \in \mathcal{M}, \forall k \in \mathcal{K}$.}
        \STATE {Given $U_{km}$, solve $\gamma$ via many-to-many matching Algorithm and output  the association variables $\mathbf{c}^{\star}$.}
        \STATE {Given $\mathbf{c}$, solve the problem in (\ref{opt6}) via MM Algorithm (lines 10-17) and output the RIS shift matrices $\boldsymbol{\Phi}_{m}^{\star} $.}
        \STATE {Given $\boldsymbol{\Phi}_{m}^{\star} $ and  $\mathbf{c}^{\star}$, solve the problem in (\ref{opt1}) via BD algorithm and output the solution $\mathbf{W}^{\star}$.}
		\STATE {According to $\mathbf{c}^{\star}, \boldsymbol{\Phi}_{m}^{\star} , \mathbf{W}^{\star}$, calculate the WSR.}
	\end{algorithmic}  
\end{algorithm}

\newpage
\subsection{Overall Algorithm and Complexity Analyses}
Based on the above discussion, we outline a detailed procedure for solving the problem (\ref{X1}) in \textit{Algorithm 4}. 

{We now proceed to analyze the complexity of  \textit{Algorithm 4}.  In \textit{step 2}, the primary computational burden of calculating $U_{km}$ lies in evaluating the bracketed term in (\ref{phik}), which has a complexity of $\mathcal{O}\left(N^{2}\right)$. Since the MM algorithm typically converges within a few iterations, we disregard the number of iterations, yielding a total complexity for \textit{step 2} of $\mathcal{O}\left(KMN^{2}\right)$.  In \textit{step 3}, the complexity of the many-to-many matching algorithm depends on the number of UEs and the length of each UE’s preference list $p_k$, with a worst-case complexity of $\mathcal{O}\left(MK\right)$. The complexity of  \textit{step 4} is $\mathcal{O}\left(MN^{2}\right)$.
In \textit{step 5}, the complexity of joint BD algorithm is primarily determined by matrix multiplication, inversion, and SVD operations. The complexity of updating $\mathbf{S}_k$ through (\ref{Qk}) is $\mathcal{O}\left((BN_t-(K-1)N_r)^{3}\right)$, while the complexity of solving for the Lagrangian dual variables $\mu_b$ can be considered negligible.  Thus, the overall complexity of the joint BD algorithm is $\mathcal{O}\left(BK(B N_t-(K-1)N_r)^{3}\right)$.
In summary, the total complexity of \textit{Algorithm 4} is given by }
\begin{equation*}
   { \mathcal{O}\left(\mathrm{max} (BK(B N_t-(K-1)N_r)^{3},KMN^{2})\right).}
\end{equation*}
{In contrast, the complexity of the alternating optimization algorithm in a conventional RIS-assisted CF system \cite{twotimescale} is given by}
\begin{equation*}
{\mathcal{O}\left(I_{\mathrm{iter}}\left( BK^2N_{t}N_{r} + KN_{r}^{3} + K^2N_{r}^{2} \right. \right.} 
{\left. \left. + I_aB^{2}K^{2}N_{t}^{2}N_{r}^{2} + I_bM^2N^2 \right) \right),}
\end{equation*}
{where $I_{\mathrm{iter}}$, $I_a$ and $I_b$ denote the number of iterations. Evidently, the proposed algorithm achieves a lower complexity.}
%$\mathcal{O}\left(I_{iter}\left ( BK^2N_{t}N_{r} + KN_{r}^{3} + K^2N_{r}^{2}+I_aB^{2}K^{2}N_{t}^{2}N_{r}^{2}+I_bM^2N^2 \right ) \right)$
% Then, we analyze the complexity of the distance-based swap matching algorithm. Note that the use of Algorithm 2 is implicit in Algorithm 3 since we want to make sure that $\boldsymbol{\Phi}_{m}$ in $U_{m}(\gamma)$ is optimal for $\gamma ( m)$. In each iteration of Algorithm 3, we consider the worst-case scenario where there are $M(M-1)$ possible swap operations, the complexity of which is approximated by $\mathcal{O}\left(M^2\right)$. Considering the optimization of the RIS phase shift matrices after each swap, the total complexity of step 2 is $\mathcal{O}\left(T_{MM}N^{2}M^2\right)$.

\begin{table*}[h!t]
\scriptsize
   \centering
   \caption{ Simulation Parameters }
   \label{tab1}
  	\resizebox{1\textwidth }{!}{     
\begin{tabular}{|c|c|c|c|}\hline
Number of transmit antennas of each AP, $N_t$&                     4&                  Number of reflecting elements of each RIS, $N$&                100\\ \hline
Number of receive antennas of each UE, $N_r$&        2&            
 Number of APs, $B$&   4\\ \hline
Number of UEs, $K$&                     6&                  
Number of RISs, $M$&                4\\ \hline
Maximum transmit power of the AP, $\mathrm{P}_{b, \max}$&                     23 dBm&    Noise Power, $\sigma^2$&                -100 dBm\\ \hline
Weighted components, $\omega_k, \forall k$&                     1&    
Initial dual variables, $\mu_b, \forall b$&                5\\ \hline
The convergence accuracy $\varepsilon $ and $\zeta$&                     $10^{-4}$&  
The convergence accuracy $\epsilon $&                $10^{-3}$\\ \hline
Path loss exponent $\alpha_{\mathrm{AU}}$&                     4&  
Path loss exponents $\alpha_{\mathrm{AR}}$ and $\alpha_{\mathrm{RU}}$&               2.2\\ \hline
The height of AP&                     10 m&  
The height of RIS&               6 m\\ \hline
The height of UE&                     1.5 m&  
{frequency $f_c$}&            {   3.5 GHz}\\ \hline
		\end{tabular}  }             
\end{table*}

\section{Performance Evaluation} \label{secfour}
This section presents simulation results under various conditions to verify the performance of the proposed RIS-assisted CF MIMO system. Considering a three-dimensional coordinate system, where APs are distributed on the four vertices of a square with a side length of $L_{\mathrm{AP}}=300$ m, RISs are uniformly distributed in a circle with a diameter of $L_{\mathrm{RIS}}=200$ m, UEs are randomly distributed inside a square with a side length of $L_{\mathrm{UE}}=100$ m. The heights of the AP, RIS, and UE are $10$ m, $6$ m, and $1.5$ m, respectively. {Unless otherwise noted, in the initial setup, $B=4$ APs, $M=4$ RISs, $K=6$ UEs, $U_{\mathrm{match}}=K/2$, and $R_{\mathrm{match}}=M/2$.} 
Each AP has $N_t=4$ antennas, each RIS has $N=100$ passive elements, and each UE has $N_r=2$ antennas. We further set the maximum AP power to $\mathrm{P}_{b, \max}=23$ dBm,  noise power to $\sigma^2=-100$ dBm, $\omega_{k}=1$, and  $\delta =0.05$.

For the channel model, we consider the same setup as in \cite{setup1, 
 setup2}. The large-scale path loss model is shown below:
\begin{equation}
  {  \mathrm{PL}=-\mathrm{PL}_{0}-10 \alpha \log _{10}\left(\frac{d}{d_{0}}\right)-20\log _{10}\left(f_c\right), }
\end{equation}
 where $\alpha$ is the
path loss exponent and $d$ is the route distance in meters, and the reference distance is set to $d_{0}=1$ m. $\mathrm{PL}_{0}$ denotes the pathloss at the distance of $1$ meter, which is set to $32.4$ dB based on the 3GPP UMi model\cite{3GPP38901}, and we set carrier frequency $f_c=3.5$ GHz. Since there are typically numerous obstacles between the AP and UE, we set $\alpha_{\mathrm{AU}}=4$. In addition, the RIS-aided link has a higher likelihood of encountering nearly free-space path loss since the RIS is usually carefully placed, thus, we set $\alpha_{\mathrm{AR}}=\alpha_{\mathrm{RU}}=2.2$. For small-scale fading channels, we consider the general model as Rician fading, which follows
\begin{equation}
    \mathbf{H}=\sqrt{\frac{\beta_{\mathrm{}}}{1+\beta_{\mathrm{}}}} \mathbf{H}^{\mathrm{LoS}}+\sqrt{\frac{1}{1+\beta_{\mathrm{}}}} \mathbf{H}^{\mathrm{NLoS}},
\end{equation}
\begin{equation}
    \mathbf{G}=\sqrt{\frac{\beta_{\mathrm{}}}{1+\beta_{\mathrm{}}}} \mathbf{G}^{\mathrm{LoS}}+\sqrt{\frac{1}{1+\beta_{\mathrm{}}}} \mathbf{G}^{\mathrm{NLoS}},
\end{equation}
where $\beta$ denotes the Rician factor, $\mathbf{H}^{\mathrm{LoS}}$ and $\mathbf{G}^{\mathrm{LoS}}$ are the LoS component, and $\mathbf{H}^{\mathrm{NLoS}}$ and $\mathbf{G}^{\mathrm{NLoS}}$ denote Rayleigh fading component. We use the UPA model at RIS, and the ULA model at AP and UE. {The RIS contains $N = N_h \times N_v$ elements, where $N_h$ and $N_v$ are the numbers of elements in the horizontal and vertical direction, respectively, and $N_h=N_v=10$.}
{We further set $\beta_{\mathrm{AR}}=\beta_{\mathrm{RU}}=3$, and $\beta_{\mathrm{AU}}\to 0 $.}
Detailed parameters can be found in Table \ref{tab1}.

It is worth noting that the following numerical results are obtained by averaging $500$ independent channel generations. To confirm the viability and efficacy of the proposed algorithm, we propose to compare it with the following benchmarks:
\begin{itemize}
    \item \textbf{Without RIS:} All APs serve UEs directly without RIS. Set all RIS phase shift matrices to zero matrices, i.e., $\mathbf{\Phi}_m=\mathbf{0}$.
    \item \textbf{Random phase:} The phases of all RIS elements are randomly generated within 0 to $2\pi$. 
    \item \textbf{Discrete phase\cite{discreteRIS}:} In practice, the phase of RIS element tends to be discrete due to limited precision. Assuming discrete phase $\theta_{m,n} \in  \mathcal{F} ,\forall m,n$, where $\mathcal{F}=\left\{0, \frac{2 \pi}{2^{i}}, \frac{2 \pi \times 2}{2^{i}}, \cdots, \frac{2 \pi \times\left(2^{i}-1\right)}{2^{i}}\right\}$ and $i$ denotes phase resolution. In this scheme, all steps are the same as in Algorithm 4 except for the projection of the optimized $\theta_{m,n}$ onto the nearest discrete value in $\mathcal{F}$.
    \item \textbf{Full association\cite{yao_robust_2023}:} Like the conventional RIS-assisted cell-free network, all RISs serve all UEs simultaneously, i.e., $c_{m,k}=1, \forall m,k$.
    \item \textbf{Direct link blocked:} The channels between all APs and all UEs are severely obstructed, i.e., $\mathbf{H}_{b,k}^d=\mathbf{0}, \forall b,k$. 
    {\item \textbf{Multicell network\cite{setup2}:} In contrast to the CF network, each AP provides service to UEs within its coverage area.}
    \item {\textbf{AO-based scheme\cite{twotimescale}}: Improving WSR in RIS-assisted cell-free network using alternating optimization.}
\end{itemize}

\begin{figure}[H]
    \centering
    \includegraphics[width=0.8\linewidth]{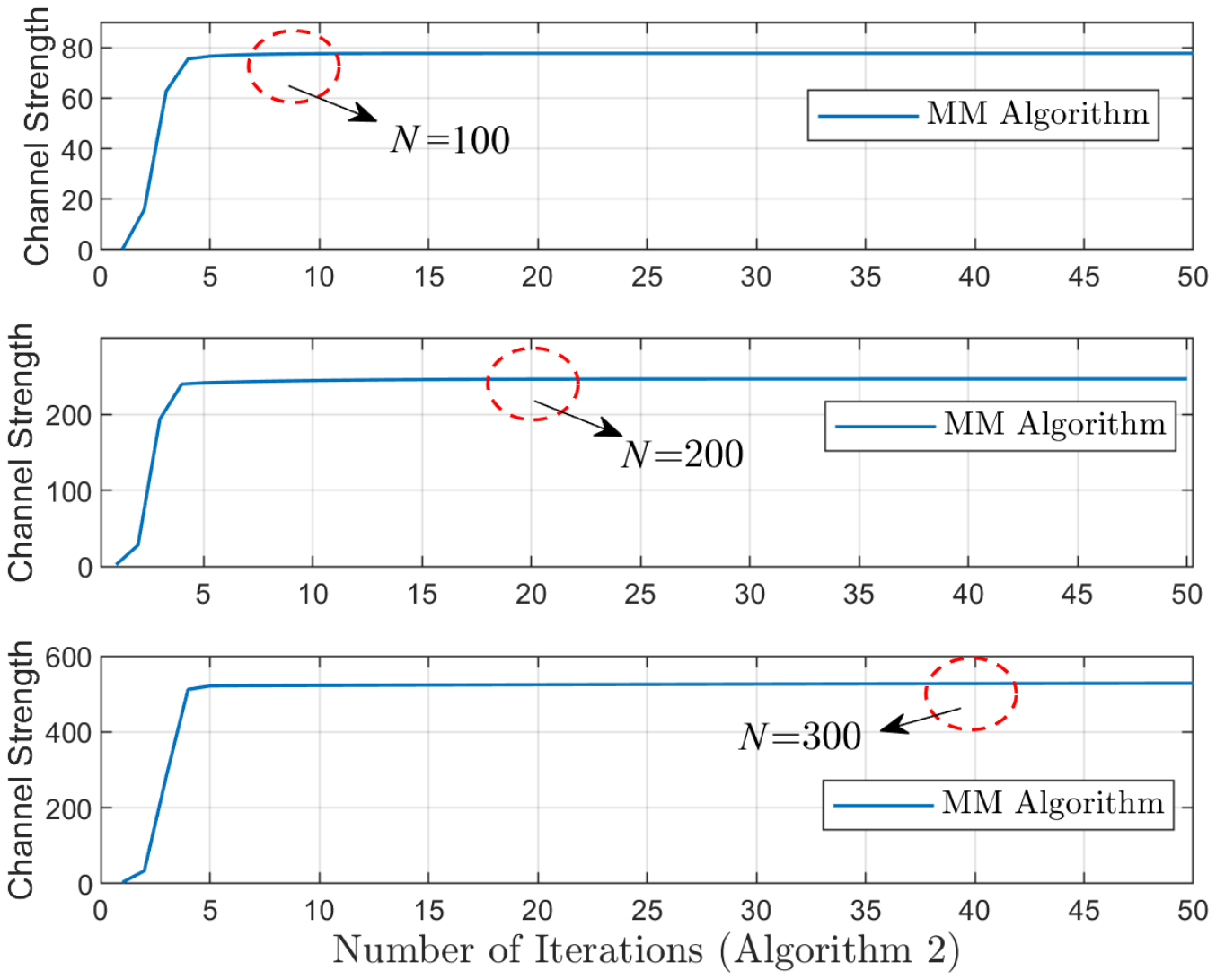}
    \caption{Convergence behavior of Algorithm 2 when $N=100, 200, 300$. }
    \label{fig4}
\end{figure}

\begin{figure}
    \centering
    \includegraphics[width=0.8\linewidth]{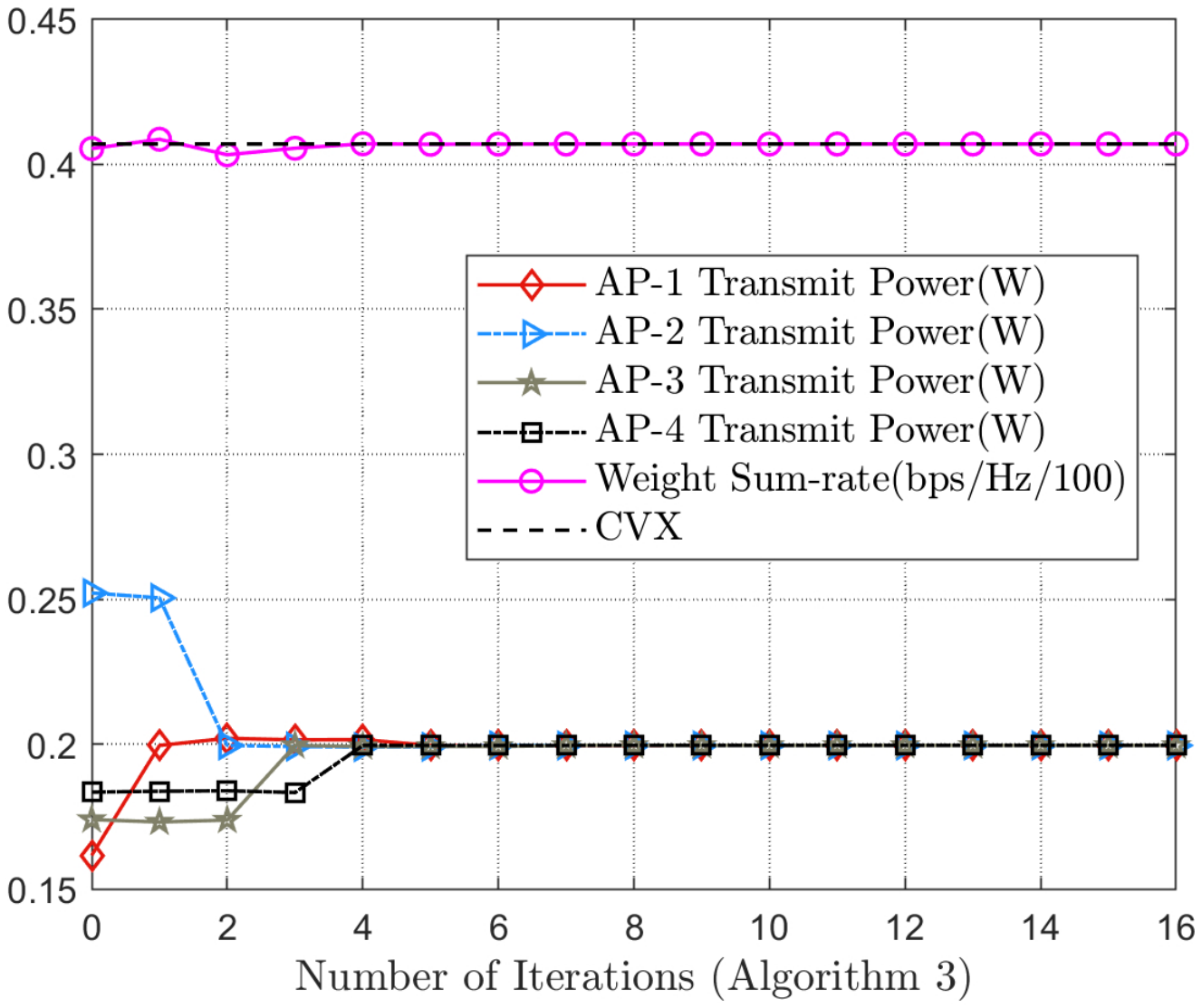}
    \caption{Convergence behavior of Algorithm 3 when $B=4$.}
    \label{fig3}
\end{figure}

\begin{table}[ht]
    \centering
    \caption{{Comparison of computing time}}
    \begin{tabular}{|c|c|c|}\hline
          & {Proposed scheme} &{ AO-based scheme\cite{twotimescale}}  \\ \hline
      {  $N_t=4, N=100$ }&{ 0.3835\,s} &{42.9506\,s } \\\hline
       { $N_t=8, N=200$} & {2.2079\,s }& {142.1823\,s } \\\hline
      {  $N_t=16, N=300$} & {10.3083\,s} &{ 256.9875\,s } \\ \hline
    \end{tabular}
\end{table}

 \subsection{Convergence Behavior and Computing Time}

 For the phase shifts optimization of the RIS, Fig. 3 illustrates the convergence performance of Algorithm 2. As expected, the number of convergence iterations of the MM algorithm increases with the number of RIS elements. However, for different values of $N$, the number of algorithm iterations does not exceed $10$. Additionally, the increase in the number of RIS elements strengthens the channel gain. 
 In Fig. 4, we show the convergence behavior of Algorithm 3. The initial $\mu_b$'s are set to $10$, $\varepsilon$ and $\zeta$ are set to $10^{-4}$. It can be observed that the AP transmit power and the weighted sum rate converge in an oscillatory manner within $10$ iterations. The converged power levels of all four APs reach the maximum allowed constraint of $23$ dBm. Notably, the algorithm achieves the same result as the CVX solver\cite{cvx}, which demonstrates the effectiveness of the algorithm.
   
{In Table II, we compare the average computation time of the proposed scheme with that of the AO-based scheme under different numbers of AP antennas and RIS elements. These simulations were conducted using MATLAB R2022b on a hardware setup with a 13th Gen Intel(R) Core(TM) i5-13600K CPU (3.50 GHz) and 32.0 GB of RAM. It can be clearly observed that the proposed scheme is more efficient, as it requires substantially less time to compute compared to the AO-based scheme across all these tested configurations, especially when the number of AP antennas and number of RIS elements becomes large.}

\subsection{Impact of Key Parameters}

\begin{figure}
    \centering
    \includegraphics[width=0.8\linewidth]{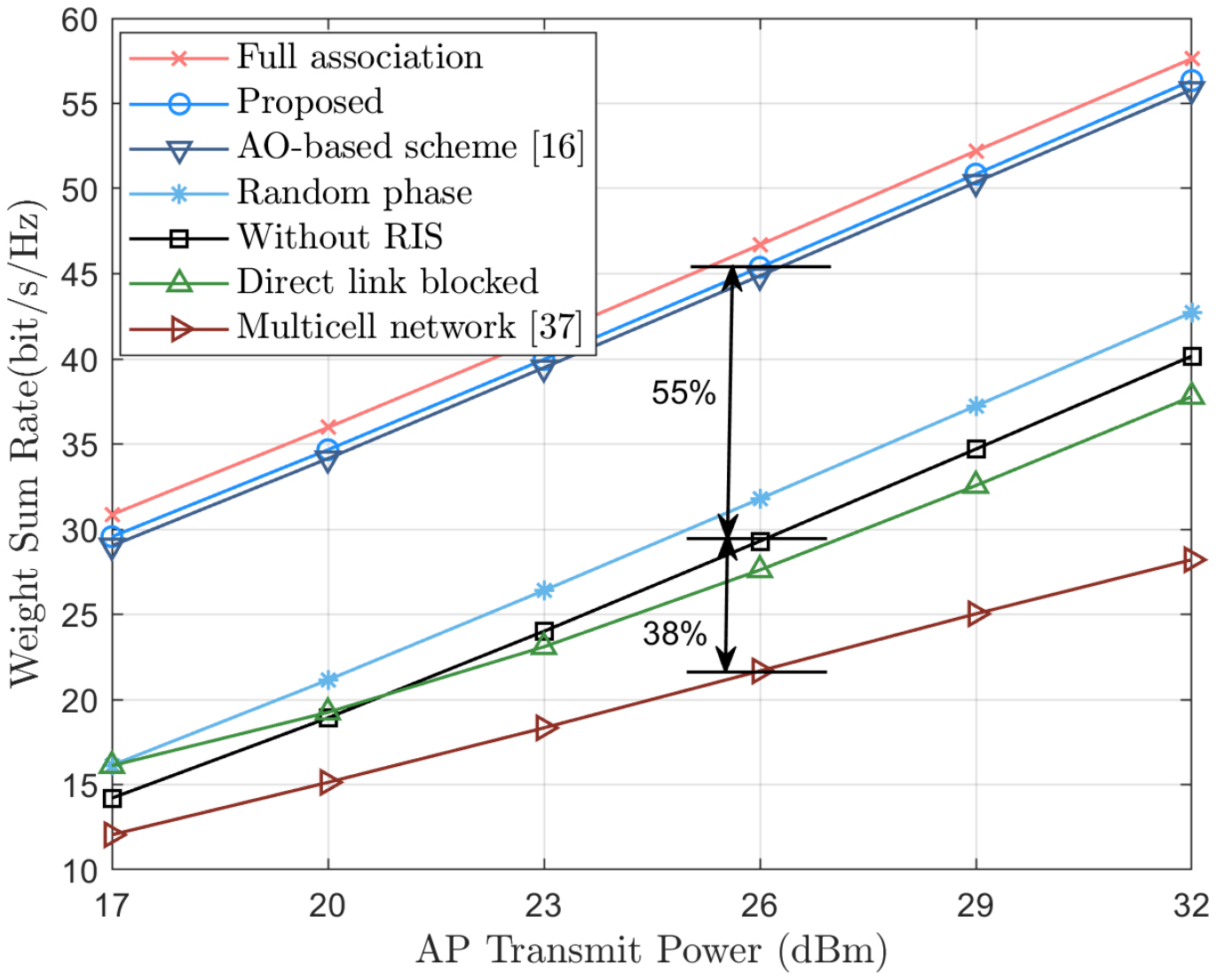}
    \caption{Weighted sum rate versus the AP transmit power $\mathrm{P}_{b, \max}$. }
    \label{AP_power}
\end{figure}

Fig. \ref{AP_power} shows a comparison of weighted sum rate for different cases of maximum AP transmit power. From Fig. \ref{AP_power}, we can make the following observations. The weighted sum rate increases rapidly as $\mathrm{P}_{b, \max}$ increases. {Particularly, the ``\textit{Proposed}" performs as well as or even better than the ``\textit{AO-based scheme}", showing that our scheme is able to achieve good performance with low complexity.} {Moreover, when $\mathrm{P}_{b, \max}=26$ dBm, the ``\textit{Without RIS}" enhances the network capacity by 38$\%$ compared to ``\textit{Multicell network}".
 This indicates that, compared to traditional multicell (or cellular) networks, the collaborative service provided by multiple APs in cell-free networks can further enhance the overall performance of the system.}
 At the same time, the ``\textit{Proposed}" improves the system performance by 55$\%$ compared to ``$\textit{Without RIS}$", demonstrating the advantages that RIS brings to the cell-free system.
 {However, introducing RIS-UE association results in some performance loss, as seen in the comparison between the “\textit{Proposed}'' and the ``\textit{Full association}''. This loss occurs because, with RIS-UE association, each RIS is restricted to serving only a subset of users, thereby limiting the reflective assistance that other users could otherwise receive. In a fully associated setup, every RIS can contribute to enhancing the signal quality for all users, maximizing the overall channel gain across the network. By contrast, RIS-UE association reduces the total effective channel gain, as users outside a specific RIS's service range cannot benefit from its reflected signals. Despite this trade-off, the performance loss remains within an acceptable range; moreover, RIS-UE association significantly reduces the channel acquisition overhead.}
% However, introducing RIS-UE association results in some performance loss, as seen in the comparison between the ``\textit{Proposed}" and the ``\textit{Full association}", but this loss remains within an acceptable range.
 
Notably, the  ``\textit{Random phase}" offers little performance gain over the ``\textit{Without RIS}", which proves the necessity of RIS phase optimization. ``\textit{Without RIS}" represents that the APs allocate all power to the AP-UE direct links, while ``\textit{Direct link blocked}" represents that the APs allocate all power to the AP-RIS-UE indirect links. It can be visualized in Fig. \ref{AP_power} that the performance of ``\textit{Direct link blocked}" becomes progressively worse than that of ``\textit{Without RIS}" when $\mathrm{P}_{b, \max}$ increases. The reason is that the equivalent path loss of the AP-RIS-UE indirect link is higher than the path loss of the AP-UE  direct link\cite{setup2}. When the AP power is high enough, the AP prefers to allocate power to the direct link, which leads to better performance contribution.

\begin{figure}
    \centering
    \includegraphics[width=0.8\linewidth]{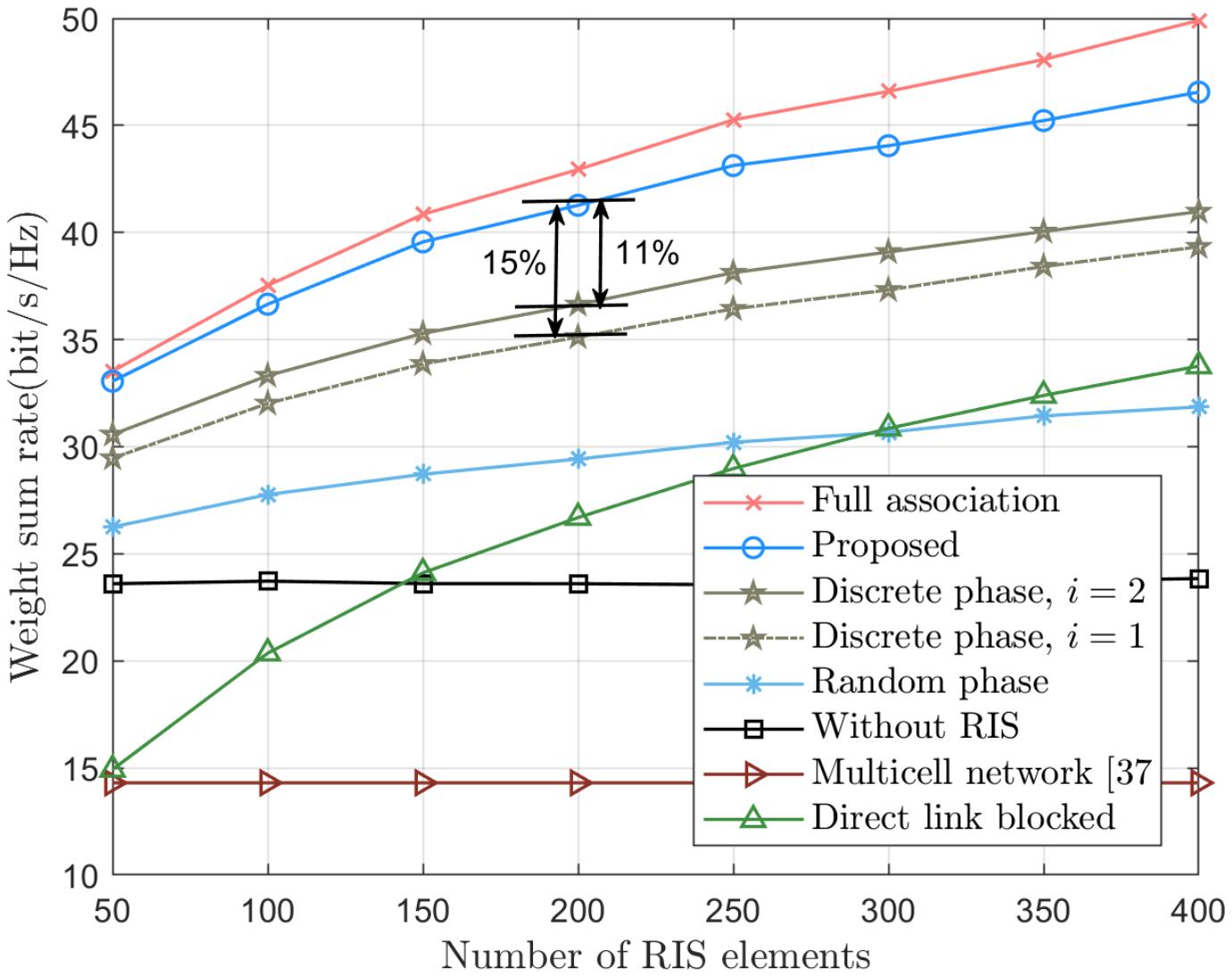}
    \caption{Weighted sum rate versus the number of RIS elements. }
    \label{RIS_element}
\end{figure}

%Fourth, when the RIS phase is discrete, performance is slightly degraded compared to continuous RIS, but, this situation can be effectively improved with increasing resolution.
%For example, when $N=$ 100, the gap is 34$\%$, however, when $N=$ 180, the gap is 48$\%$. This illustrates the necessity that RIS should only optimize the phase for a specific UE while reducing unnecessary channel estimation overheads.
In Fig. \ref{RIS_element}, we investigate the effect of the number of RIS elements on the achievable weighted sum rate. Undoubtedly, the increase of RIS elements will make the system performance better. On the other hand, as the number of RIS elements increases, the gap between ``\textit{Full association}" and ``\textit{Proposed}" becomes progressively larger.  This is because an increase in the number of RIS elements introduces more channels, resulting in greater performance loss due to RIS-UE association. 
When the RIS phase is discrete, performance is slightly degraded compared to continuous RIS, however, this degradation can be effectively reduced with higher phase resolution.
When $N=200$ and the phase resolution $i=1$, the system performance is lost by 15\%, while when the phase resolution $i=2$, the system performance is only lost by 11\%.
What is more, the performance loss increases with the number of RIS elements as the phase precision of the RIS decreases, which can be verified by comparing the curves ``\textit{Discrete phase, $i=2$}" and ``\textit{Discrete phase, $i=1$}". However, this can reduce hardware and power overhead.  
It can also be seen that curve ``\textit{Direct link blocked}" gradually exceeds curve ``\textit{Without RIS}" as the number of RIS elements increases, which is due to the fact that the gain of the reflected channel gradually becomes larger.

\begin{figure}
    \centering
    \includegraphics[width=0.8\linewidth]{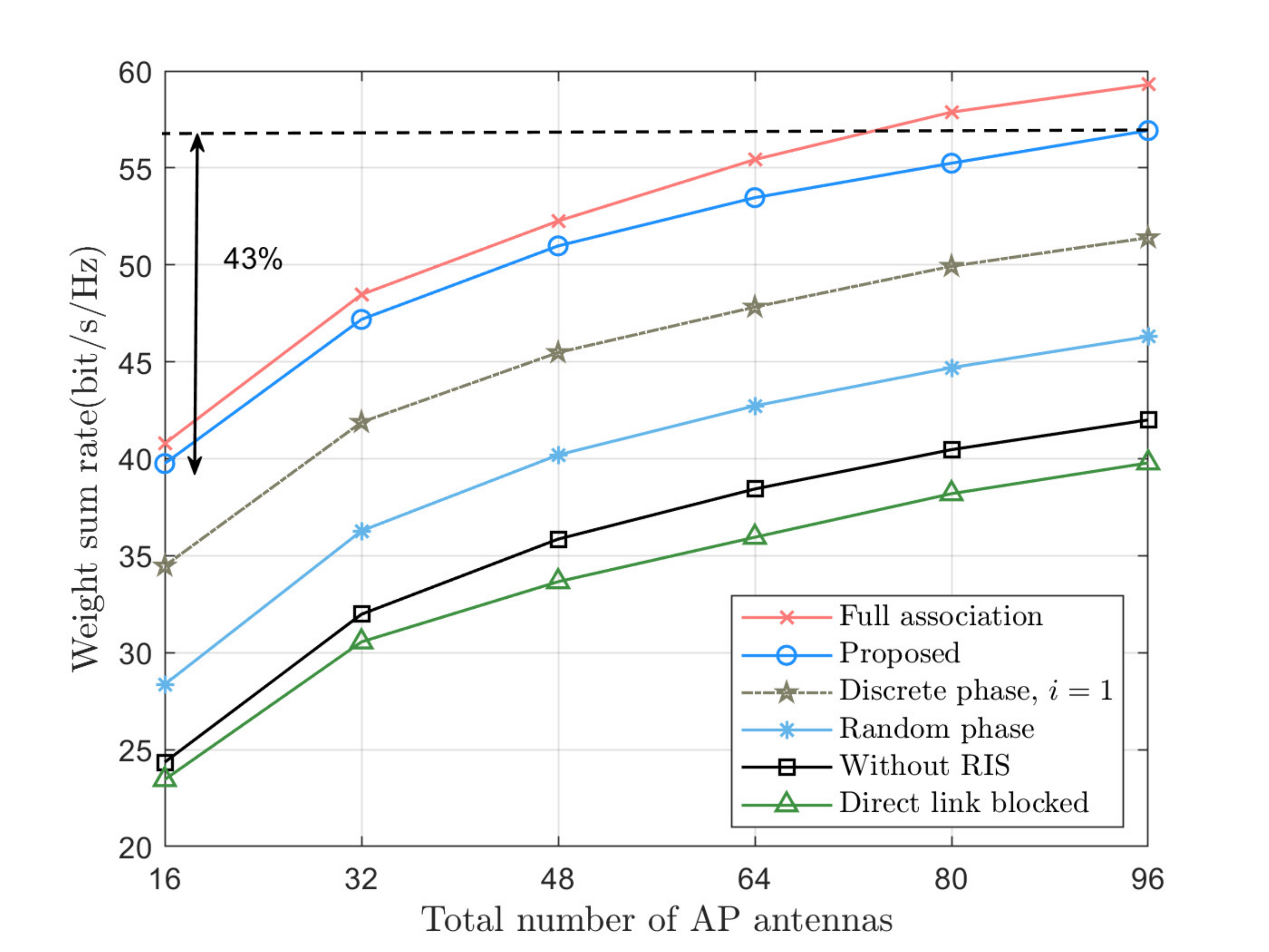}
    \caption{Weighted sum rate versus AP antennas.}
    \label{AP_antenna}
\end{figure}

Fig. \ref{AP_antenna} illustrates the correlation between the weighted sum rate and the total number of AP antennas, where the total number of AP antennas is equal to $BN_t$. 
From Fig. 7, we see that all the curves increase with the number of AP antennas, since more antennas lead to more spatial gain. When the total number of AP antennas increases from 16 to 96, the ``\textit{Proposed}" demonstrated a 43\% improvement in performance. Moreover, as the total number of AP antennas increases, the performance gain of the system begins to diminish. This is likely due to the law of diminishing returns, where each additional antenna contributes less to overall performance improvement as the system approaches its optimal capacity. 

\subsection{Impact of RIS Deployment}

\begin{figure}
    \centering
    \includegraphics[width=0.8\linewidth]{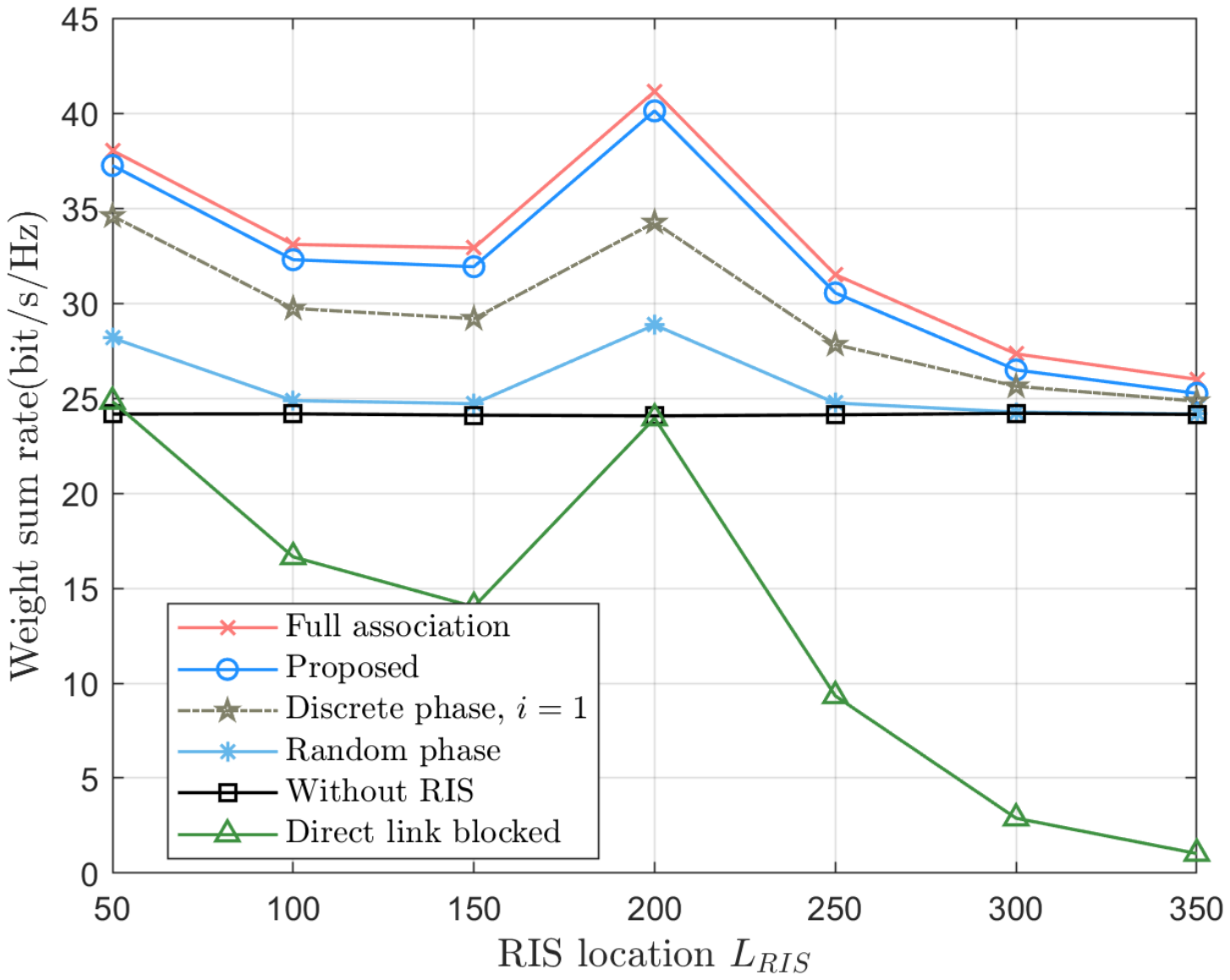}
    \caption{Weighted sum rate versus the RIS deployment. }
    \label{RIS_distance}
\end{figure}

{In Fig. \ref{RIS_distance}, we explore the impact of RIS deployment on system performance,  where the horizontal coordinate is the diameter $L_{\mathrm{RIS}}$ of the circle in which the RISs are located. Two peaks are visible: one at $50$ m and another at $200$ m, especially noticeable at the line ``\textit{Direct link blocked}". The initial peak occurs when the RISs are in proximity to the UEs, and the second peak occurs when the RISs are in the middle of APs and UEs. When RISs are far away from APs and UEs, i.e., $L_{RIS}>300 $ m, the performance of the system drops drastically. Therefore, it can be concluded that the RISs are best deployed in the middle of the APs and the UEs. Additionally, the RISs should be positioned closer to the UEs, which will result in greater performance gains. Furthermore, we specifically note that when the RISs are close to the UEs, ``\textit{Random phase}" also provides a partial performance gain compared to ``\textit{Without RIS}". This may be due to the fact that the UEs are able to receive stronger reflected signals from the RISs even when the RIS phase shifts are not optimized.}

\subsection{CDF of Weight Sum Rate}

\begin{figure}
    \centering
    \includegraphics[width=0.8\linewidth]{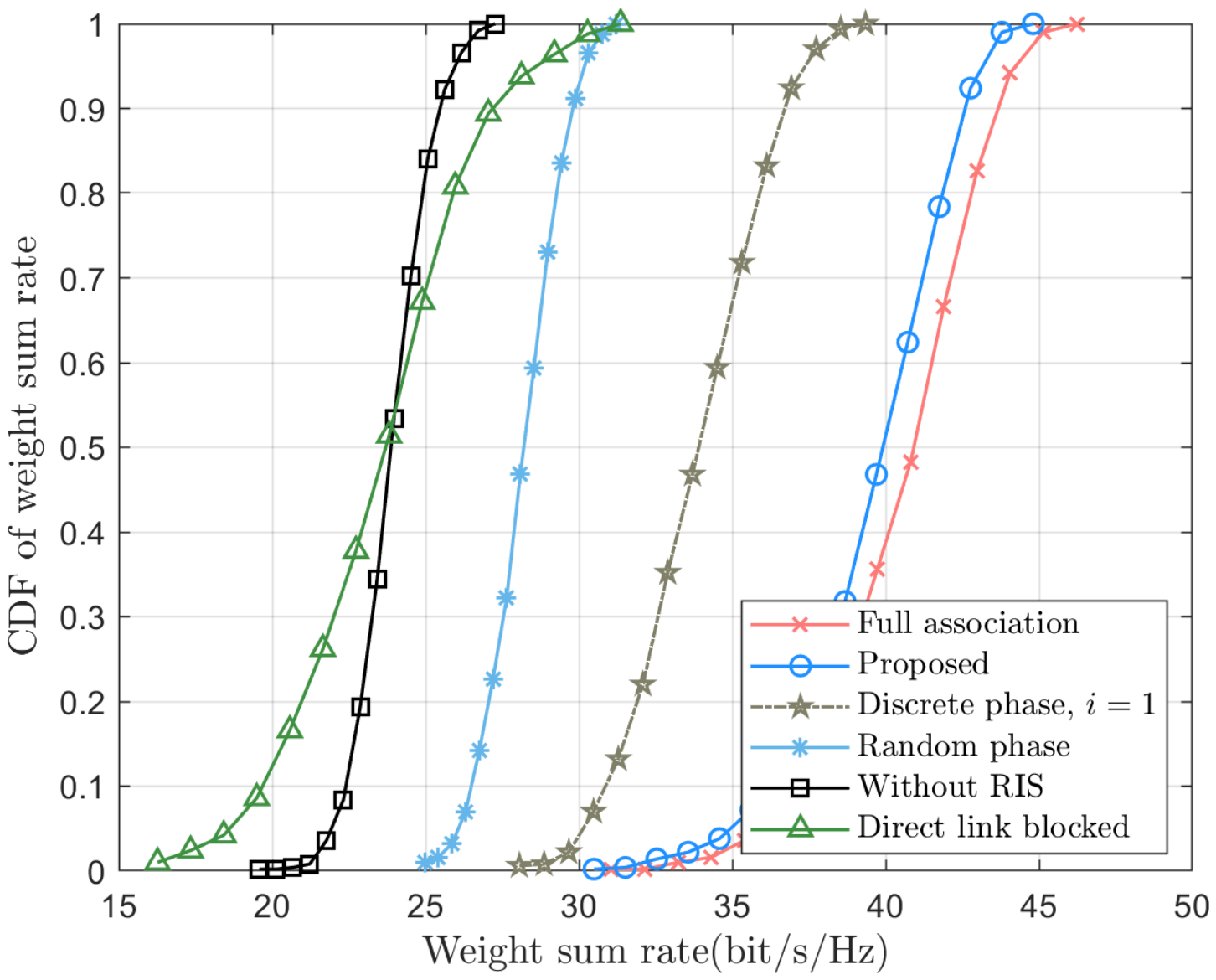}
    \caption{{CDF of weight sum rate.}}
    \label{fig8}
\end{figure}

In Fig. \ref{fig8}, we show the  cumulative distribution function (CDF) of the weighted sum rate. It can be seen that the ``\textit{Proposed}" has a $78.4$\% probability of achieving a weighted sum rate of at least $41.7$ bps/Hz, which is slightly lower compared to the ``\textit{Full association}". This illustrates that the RIS-UE association does not lead to significant performance loss for the RIS-assisted CF network. Moreover, the maximum weight sum rate of the ``\textit{Proposed}" is about $162$\% of that of the ``\textit{Without RIS}", demonstrating the benefits of combining RIS with CF-MIMO. However, in systems with a large number of RIS elements, obtaining complete CSI, including both LoS and NLoS paths in Rician channels, is often impractical. In such cases, LoS path information can be inferred from UEs' locations, which facilitates transmission optimization. The ``\textit{Proposed, LoS-only, $\beta_{\mathrm{AR}}=\beta_{\mathrm{RU}}=3$}'' curve illustrates the algorithm's performance when only LoS path information in the AP-RIS and RIS-UE channels is available, while NLoS path information remains unknown. Compared to the ``\textit{Proposed}'', which assumes full knowledge of both LoS and NLoS paths, the LoS-only scheme shows a performance degradation of 1 to 2 bps/Hz, which remains within an acceptable range. Furthermore, as $\beta_{\mathrm{AR}}$ and $\beta_{\mathrm{RU}}$ increase, indicating a higher proportion of LoS paths, the performance gap between the LoS-only and complete CSI-aware schemes narrows. This suggests that, in the absence of complete CSI, an increased proportion of LoS paths enhances the algorithm's performance even when only user locations are known.

\subsection{Robustness to imperfect CSI}

\begin{figure}
    \centering
    \includegraphics[width=0.8\linewidth]{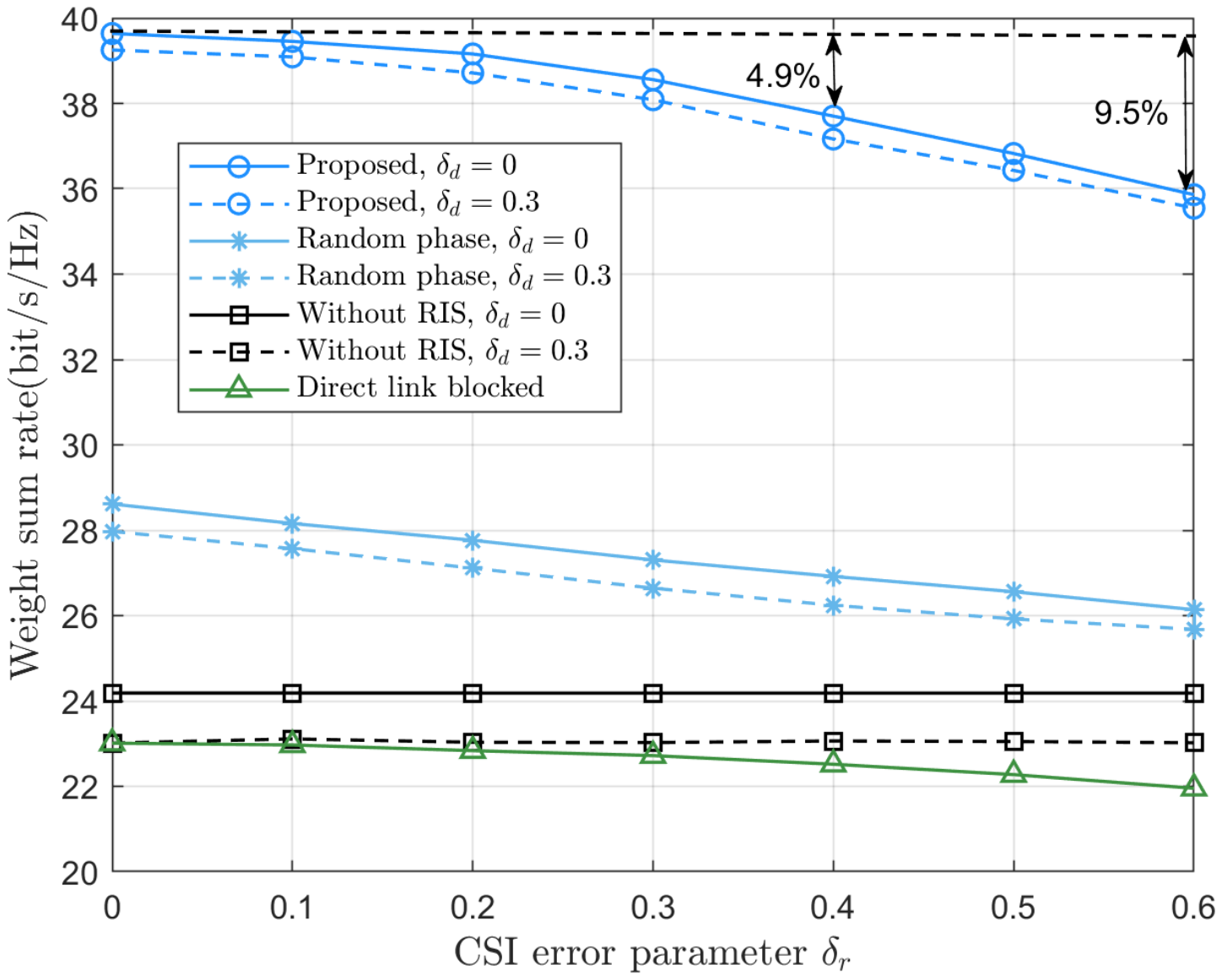}
    \caption{{Weighted sum rate against the CSI error parameter $\delta_r$ when $\delta_d=\left \{ 0,0.3 \right \} $.}}
    \label{CSI_error}
\end{figure}

{In the RIS-assisted CF system, it is difficult to obtain perfect CSI due to the presence of a large number of antennas. Therefore, it is necessary to verify the robustness of the designed scheme against imperfect CSI. The estimated channel is modeled as $\hat{h}=h +e $ \cite{twotimescale,imperfectCSI}, where $h$ denotes the perfect channel and $e$ signifies the estimation error, which follows a Gaussian distribution with zero mean, i.e., $e \sim \mathcal{CN}\left(0,\sigma_{e}^2\right)$. We define $\sigma_{e}^2=\delta \left | h \right | ^2$, where $\delta$ represents the ratio of the estimation error to the channel gain $\left | h \right | ^2$. Furthermore, $\delta_r$ is the error ratio for the cascaded channel (AP-RIS channel and RIS-UE channel) and $\delta_d$  is the error ratio for the direct channel (AP-UE channel). In Fig. \ref{CSI_error}, we plot the WSR against the CSI error parameter $\delta_r$. It is observed that as $\delta_r$ increases, as expected, the performance of the proposed scheme is affected. When $\delta_r=0.4$, the proposed method experiences a $4.9$\% degradation in terms of the WSR, and at $\delta_r=0.6$, the WSR degradation expands to $9.5$\%. Additionally, when $\delta_d$ changes from $0$ to $0.3$, the system experiences only a slight decrease in performance. This indicates that our solution demonstrates strong robustness against small-to-moderate CSI errors which usually occur in current wireless networks\cite{CSI_error3}. }

\subsection{Impact of the Weights}

\begin{figure}[t]
    \centering
     \includegraphics[width=0.8\linewidth]{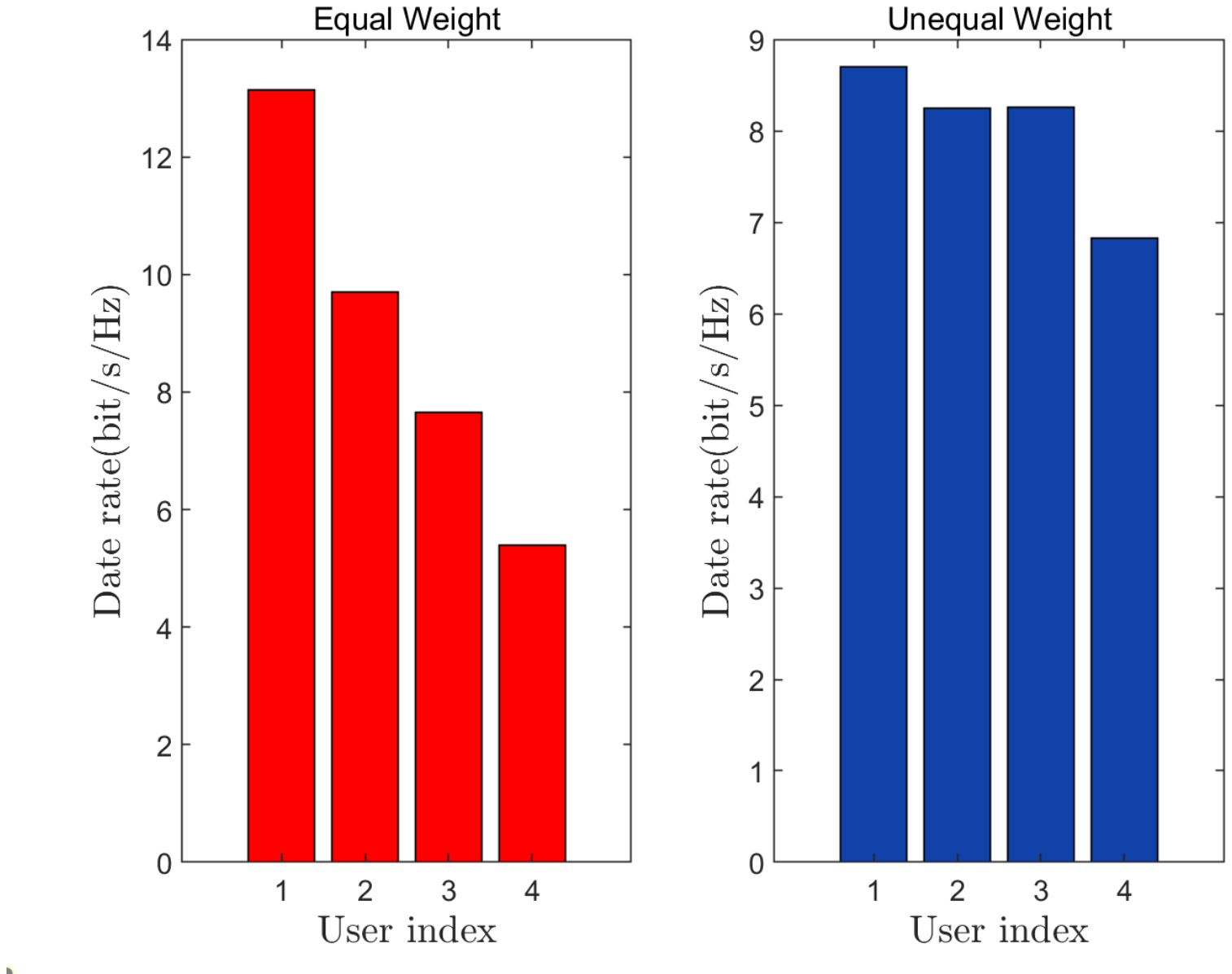}
    \caption{{ Individual data rate under two sets of weights.}}
    \label{weights}
\end{figure}

{Weight factor $\omega$ can be used to control fairness among users. To illustrate this better, we provide an example of a pure cell-free system (without RIS). We set $B=2$ and $K=4$. Two APs are located at $(-30,0,10)$ and $(30,0,10)$, while four users are positioned at $(75,0,1.5)$, $(-100,0,1.5)$, $(150,0,1.5)$, and $(-175,0,1.5)$, with indices $1$, $2$, $3$, and $4$, respectively. Fig. \ref{weights} depicts the individual data rates under two sets of weight factors. For equal weights, we set $\omega=0.5$ for each UE. For unequal weights, we set $\omega_{1}=0.1$, $\omega_{2}=0.25$, $\omega_{3}=0.65$, and $\omega_{4}=0.9$. It can be observed that when the weights are equal, UEs closer to the AP achieve higher data rates, and there is a significant difference in rates among UEs. By allocating greater weight factors to UEs farther away from the AP, we can promote fairness among UEs and attain a more equitable distribution of data rates.}

\section{Conclusion}\label{secfive}   
In this article, we address RIS-UE association in the design of a RIS-assisted downlink CF MIMO network. By jointly optimizing active and passive beamforming along with RIS-UE association variables, we aim to maximize the WSR. To tackle this complex MINLP problem, we propose a two-stage framework comprising a many-to-many matching algorithm for RIS-UE association, an MM algorithm for RIS phase shift optimization, and a joint BD algorithm based on the bisection method for AP beamforming. Simulation results demonstrate the efficacy of the proposed algorithm across diverse scenarios and underscore the essential role of RIS-UE association in balancing performance and channel acquisition overhead within cell-free networks.
In future work, in addition to considering the RIS-UE association, we should also examine the impact of the AP-UE association on the system to reduce the overhead caused by APs exchanging information with each other.
\section*{Appendix}
\section*{Proof of Lemma 1}
 We expand $\mathbf{x}^\mathrm{\mathit{H} } \mathbf{L} \mathbf{x}$ and write it in the form:
 \begin{align}
 \label{proof3.1} \mathbf{x}^\mathrm{\mathit{H} } \mathbf{L} \mathbf{x}    &=  (\mathbf{x} - \mathbf{x}_{(t)} + \mathbf{x}_{(t)})^\mathrm{\mathit{H} } \mathbf{L} (\mathbf{x} - \mathbf{x}_{(t)} + \mathbf{x}_{(t)})   \notag \\
& = (\mathbf{x} - \mathbf{x}_{(t)})^\mathrm{\mathit{H} } \mathbf{L} (\mathbf{x} - \mathbf{x}_{(t)}) + (\mathbf{x} - \mathbf{x}_{(t)})^\mathrm{\mathit{H} } \mathbf{L} \mathbf{x}_{(t)}  + \mathbf{x}_{(t)}^\mathrm{\mathit{H} } \mathbf{L} (\mathbf{x} - \mathbf{x}_{(t)}) + \mathbf{x}_{(t)}^\mathrm{\mathit{H} }  \mathbf{L} \mathbf{x}_{(t)}  ,
 \end{align}
since $\mathbf{L}$ is a semi-positive definite matrix, we have
\begin{equation}
 \label{proof3.2}  (\mathbf{x} - \mathbf{x}_{(t)})^\mathrm{\mathit{H} } \mathbf{L} (\mathbf{x} - \mathbf{x}_{(t)}) \geq 0 . 
\end{equation}
Substituting (\ref{proof3.2}) into (\ref{proof3.1}), we can deduce that (\ref{lemma3}) in Lemma 3 is true. 
%\begin{thebibliography}{1}
\section*{Proof of Lemma 3}
 We first define a function $g_{b}\left(\mu_{b}\right)=\left ( \mu_{b} \mathbf{Z}_{b, k}^{1}+\mathbf{Z}_{b, k}^{2} \right ) ^{-\frac{1}{2}}$, thus the (\ref{f_mu_b}) can be re-expressed as 
\begin{equation}
    f_{b}\left(\mu_{b}\right)=\operatorname{Tr}\left ( { \sum_{k=1}^{K}} \mathbf{Z}_{b, k}^{1} g_{b}\left(\mu_{b}\right) \tilde{\mathbf{S}}_{k} g_{b}\left(\mu_{b}\right) \right ) .
\end{equation}
From the basic derivative chain rule and the fundamental properties of matrix trace, the derivative of $f_{b}\left(\mu_{b}\right)$ with respect to $\mu_{b}$ is
\begin{small}
\begin{align}
 \frac{\mathrm{d}  f_{b}\left(\mu_{b}\right)}{\mathrm{d} \mu_{b}} &=\operatorname{Tr}\left(\frac{\mathrm{d}}{\mathrm{d} \mu_{b}}\left(\sum_{k=1}^{K} \mathbf{Z}_{b, k}^{1} g_{b}\left(\mu_{b}\right) \tilde{\mathbf{S}}_{k} g_{b}\left(\mu_{b}\right)\right)\right) \notag\\
&=\operatorname{Tr}\left(\sum_{k=1}^{K} \left(\mathbf{Z}_{b, k}^{1} \frac{\mathrm{d} g_{b}\left(\mu_{b}\right)}{\mathrm{d} \mu_{b}} \tilde{\mathbf{S}}_{k} g_{b}\left(\mu_{b}\right) 
 + \mathbf{Z}_{b, k}^{1} g_{b}\left(\mu_{b}\right) \tilde{\mathbf{S}}_{k} \frac{\mathrm{d} g_{b}\left(\mu_{b}\right)}{\mathrm{d} \mu_{b}}\right)\right). \label{df}
\end{align}
\end{small}
The derivative of $g_{b}\left(\mu_{b}\right)$ with respect to $\mu_{b}$, which can be formulated as
\begin{equation}
  \label{dg}   \frac{\mathrm{d}  g_{b}\left(\mu_{b}\right)}{\mathrm{d} \mu_{b}}=-\frac{1}{2}\left(\mu_{b} \mathbf{Z}_{b, k}^{1}+\mathbf{Z}_{b, k}^{2}\right)^{-\frac{3}{2}} \mathbf{Z}_{b, k}^{1}. 
\end{equation}
It is easy to show that both $\mathbf{Z}_{b, k}^{1}$ and $\left(\mu_{b} \mathbf{Z}_{b, k}^{1}+\mathbf{Z}_{b, k}^{2}\right)$  are positive definite matrices. While we substitute (\ref{dg}) into (\ref{df}), it is evident from the properties of positive definite matrices that $\frac{\mathrm{d} f_{b}\left(\mu_{b}\right)}{\mathrm{d} \mu_{b}} < 0$  when $\mu_{b} \ge 0$.

\bibliographystyle{IEEEtran}
\bibliography{IEEEabrv,mylib}

\vfill
                           
\end{document}